\documentclass{article}
\usepackage{amssymb,amsmath,amsthm}
\usepackage[left=.75 in, right=.75 in,top=.75 in, bottom=.75 in]{geometry}
\usepackage{dirtytalk}
\usepackage{graphicx}
\usepackage{comment}
\usepackage{subcaption}
\usepackage{titling}
\usepackage{hyperref}
    \hypersetup{
        colorlinks=true,
        linkcolor=black,
        citecolor=black,
        urlcolor=black,
        filecolor=black,
        anchorcolor=black
    }

\usepackage{pdflscape}

\usepackage{booktabs}
\usepackage{mdframed}
\usepackage{blindtext}
\usepackage{titlesec}
\usepackage{lipsum} % For dummy text
\usepackage{transparent}
\usepackage{tocloft}
\usepackage{tikz}
\usepackage{tabularx}
\usepackage{algorithm}
\usepackage{algpseudocode}

\algnewcommand{\algorithmicparfor}{\textbf{par for}}

\algdef{SE}[PARFOR]{ParFor}{EndParFor}[1]
  {\algorithmicparfor\ #1}
  {}

\algtext*{EndParFor}
\usepackage{longtable}
\usetikzlibrary{calc}
\usepackage{diagbox}
\def\H1{{_0}H^1(\Omega)}

\def\a{\mathcal{A}}

\newcommand{\C}{\mathcal{C}}

\def\XXint#1#2#3{{\setbox0=\hbox{$#1{#2#3}{\int}$ }
\vcenter{\hbox{$#2#3$ }}\kern-.6\wd0}}

\newtheorem{lem}{Lemma}[section]
\newtheorem{cor}[lem]{Corollary}

\newtheorem{thm}[lem]{Theorem}

\newtheorem*{cor*}{Corollary}
\newtheorem*{prop*}{Proposition}
\newtheorem*{thm*}{Theorem}
\newtheorem*{definition*}{Definition}
\newtheorem*{lem*}{Lemma}

\begin{document}

\title{The Optimal Knight Exchange Puzzle is NP-Hard}
\author{Henry Siegel}
\date{}

\maketitle
\section*{Abstract}
This paper explores the hardness of two popular recreational chess puzzles: The Knight's Tour and the Knight Exchange (Swap). The problem of finding a knight's tour is known to be NP-hard for any chessboard with holes and constant-time decidable for rectangular chessboards, so a natural direction is to explore the hardness of the problem for intermediate chessboard restrictions. In this paper, we show that Knight's Tour is NP-hard for connected boards. We also give a short polynomial-time reduction between the two problems, showing that the optimality version of Knight Exchange is NP-hard.

\section*{Introduction}
Knight's Tour and Knight Exchange are popular recreational puzzles that have appeared in video games and online puzzle forums. 

\begin{itemize}
    \item The Knight's Tour puzzle: Given a chessboard, can a knight visit every square exactly once (A knight moves in an L shape as prescribed by FIDE chess rules). In the open variant, the start and end squares are different. In the closed variant, the knight returns to its starting square. 
    \item The Knight Exchange puzzle: Given a chessboard with pairs of black and white knights, can you swap the black and white knights with a sequence of moves so that every square contains at most one knight at any time? ~\ref{fig:eleventh_hour_puzzle}
\end{itemize}

\begin{figure}[htbp]
    \centering
    % Left Image
    \begin{subfigure}[b]{0.285\textwidth}
        \centering
        \includegraphics[width=\textwidth]{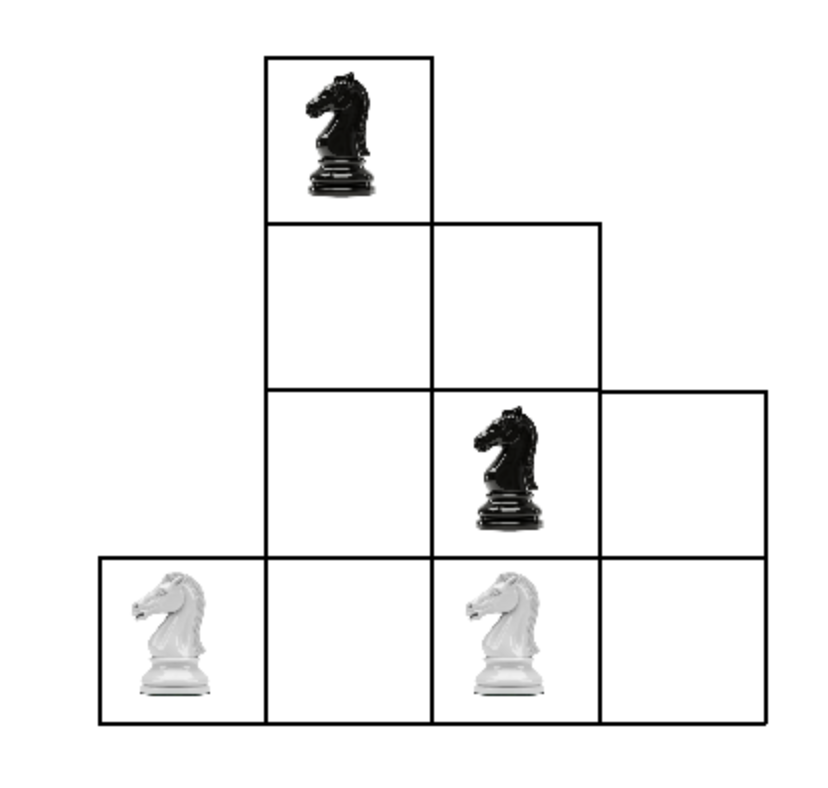}
        \caption{Starting position}
    \end{subfigure}
    % Right Image
    \begin{subfigure}[b]{0.285\textwidth}
        \centering
        \includegraphics[width=\textwidth]{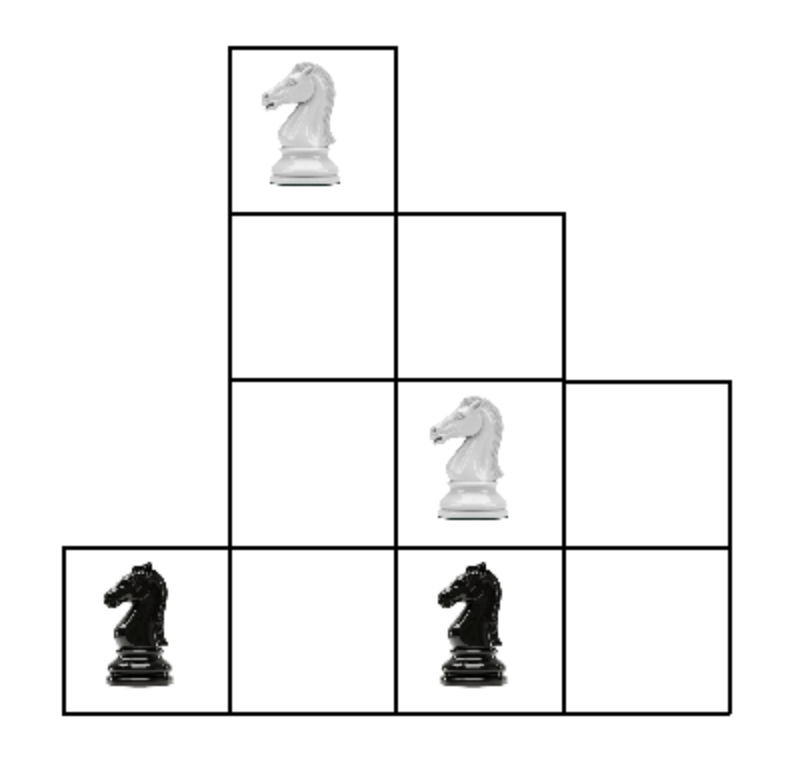}
        \caption{knights swapped places}
    \end{subfigure}
    
    \caption{This particular puzzle appeared in the video game \textit{The Eleventh Hour}}
    \label{fig:eleventh_hour_puzzle}
\end{figure}
\pagebreak

\subsection*{Definitions}
\begin{itemize}
    \item A chessboard with holes is a square boolean matrix where a 1 is a square and a 0 is the absence of a square.
    \item A connected chessboard is a chessboard with holes such that a rook (as prescribed by FIDE chess rules) can move between any two squares in finite steps.
    \item A knight's graph is a simple undirected bipartite graph where the vertices are the squares of a chessboard. Two vertices are connected if a knight can move between them in one step.
    \item An Open Knight's Tour is a Hamiltonian Path on the knight's graph.
    \item A Closed Knight's Tour is a Hamiltonian Cycle on the knight's graph.
\end{itemize}

\subsection*{Prior Work}
The hardness of deciding whether an Open Knight's Tour exists varies on the class of input chessboard. Schwenk \cite{Schwenk1991} gave a complete classification on the dimensions of all rectangular chessboards that contain an Open Knight's Tour. Conrad \cite{ConradHindrichsMorsyWegener1994} showed a linear-time solution to compute an Open Knight's Tour for rectangular chessboards. However, for chessboards with holes, McGown et. al \cite{McGownLeininger2002} showed that the problem is NP-hard. In this paper, we will study an intermediate case, confirming a conjecture posed by McGown that the Open Knight's Tour problem is NP-hard for connected chessboards with holes. We will also show that the closed variant is NP-hard for the same types of chessboard.

Goraly and Hassin \cite{GoralyHassin2010} formulated the Knight Exchange Puzzle as an instance of a multicolored pebble motion problem. Pebble motion problems are characterized by an undirected simple graph with $n$ vertices with $p \leq n$ pebbles, where each vertex can hold at most one pebble. A move consists of transfering a pebble from an occupied vertex to an adjacent unoccopied vertex. The feasibility version is whether two configurations of pebbles can be transformed through a series of moves, and the optimality version cares about the length of the solution. Versions of the problem have been studied like \cite{KornhauserMillerSpirakis1984} and \cite{calinescu2005}, where the pebbles are indistinguishable or distinguishable (by color) or with graph restrictions such as trees, planar graphs, or grid graphs. 

Viewed as a combinatorial optimization problem, Iordan \cite{Iordan2019} computed an optimal solution on a $4n \times 4n$ chessboard with $2n$ black and white knights placed on the first and last rows. Iranpoor \cite{Iranpoor2021} formulated the puzzle as a Network Flow problem when the board is fixed to be a $3 \times 3$ square with the black and white knights on the first and last rows. 

From a hardness perspective, Goldreich \cite{Goldreich} proved that the $N \times N$ puzzle is NP-hard, which is an instance of the fully labeled pebble motion problem (Every pebble is a different color). Recently, Geft showed that the optimality version of multicolored pebble motion problem for two colors is NP-hard when restricted to trees. \cite{Geft2026}. The feasibility version of the Knight Exchange can be decided in poly-time as it follows from the work done by Goraly and Hassin \cite{GoralyHassin2010}. However, to the best of my knowledge, no prior work studied the hardness of the optimality version of the Knight Exchange with general boards and placement of knights. This question fits nicely in the general theory of pebble motion problems because the Knight Exchange is an instance of the two-color pebble motion problem, where the color classes are the same size, the initial and final configurations are reflected, and the graph inputs are knight's graphs. We show that the optimality version of the Knight Exchange is NP-hard.

We will also show a simple gadget-style polynomial reduction from the Knight's Tour to the Knight Exchange. The reduction can be presented outside of a chess context as it is a reduction from the Hamiltonian Cycle problem on bipartite graphs to the optimality version of the two-color pebble motion problem where the goal is to swap the colors. For brevity, we will refer to this problem as Pebble Swap. In the next section, we will present the reduction in the context of Hamiltonian Cycles and Pebble Swap because it can be adapted to their chess counterparts. The reduction can be viewed as an alternative route to showing that certain families of pebble motion problems are NP-hard. The standard routes in the literature are described in \cite{KornhauserMillerSpirakis1984} \cite{Goldreich}.

\subsection*{Results in This Paper}
\begin{itemize}
    \item The problem of determining whether a Knight's Tour exists on connected chessboards is NP-hard for both the closed and open versions.
    \item The Optimal Knight Exchange problem on connected chessboards is NP-hard.
\end{itemize}
For brevity, we will refer to the problems as Open/Closed Knight's Tour and Knight Exchange.

\section*{Pebble Swap on Bipartite Graphs}

We start with a formal definition of the problem.

    \begin{itemize}
        \item Given a graph $G = (V,E)$, a configuration is a function $\C: V \rightarrow \{-1,0,1\}$ that can be thought of as the pebble assignment function where vertices mapped to $-1$ are occupied by a black pebble, vertices mapped to $1$ are occupied by a white pebble, and vertices mapped to $0$ are empty.
        
        \item A move is the ordered pair of configurations $(\mathcal{C}_1, \mathcal{C}_2)$ such that there exists $uv \in E$ such that $\C_1 \equiv \C_2$ on $V \setminus \{u,v\}$, $\C_1(v) = \C_2(u) = 0$, and $\C_1(u) = \C_2(v) \neq 0$. For brevity, we will denote a move as $u \rightarrow v$, where the pebble on $u$ moves to $v$.

        \item A plan is a sequence of moves of the form $\langle(\C_0,\C_1), (\C_1,\C_2), \cdots, (\C_{k-2}, \C_{k-1}),(\C_{k-1}, \C_k) \rangle$. We say that the plan transforms $\C_0$ to $\C_k$.
    \end{itemize}
    Pebble Swap Problem: Given a graph $G$, a configuration $\C$ and a natural number $k$, is there a plan of length $k$ that transforms $\C$ into $\C'$ so that $\C \equiv -\C'$ (the white and black pebbles have swapped places).

The Hamiltonian Cycle problem is a classic NP-hard problem. Demaine and Rudoy \cite{DemaineRudoy2017} showed that the problem restricted to bipartite graph is still NP-Hard. We will provide a polynomial time reduction from this problem to the Pebble Swap Problem.
\noindent

\begin{thm}
    The Pebble Swap Problem on Bipartite Graphs is NP-Hard.
\end{thm}

\begin{proof}
Suppose the input graph is $G = (V,E)$. 

We will construct the modified graph $G' = (V',E')$ as follows: 
Set the new vertex set $V'$ to include all the original vertices with two extra vertices: $V' = V \cup \{x, y\}$. Fix an arbitrary $v_0 \in V^+$ and add edges $v_0 x$ and $v_0 y$. The new edge set $E'$ contains all the original edges in $E$ in addition to these two extra edges:

Let $V^-, V^+$ be the two bipartite sets of $V$. Place black pebbles on all the vertices in $V^-$ and white pebbles on all the vertices in $V^+$.
\[\C(v) = \begin{cases}1 \text{ if } v \in V^+ \\ -1 \text{ if } v \in V^- \\
0 \text{ if } v \in \{x,y\} \end{cases}\]

The reduction is $G \rightarrow (G', \C,n+4)$. There are $O(n^2)$ edges to construct, and $O(n)$ vertices to construct in $G'$, so the reduction is polynomial time.

\subsubsection*{Hamiltonian Cycle $\implies$ A plan of n+4 moves to swap the pebbles.}

Suppose $v_0 v_1 v_2 \cdots v_{n-1} v_0$ is a Hamiltonian cycle in the input graph. Recall that $v_0$ is the \say{special} vertex with edges $v_0x$ and $v_0y$. Table \ref{Hamiltonian_cycle_moves} shows a plan of $n+4$ moves to swap the black and white pebbles on $G'$. For clarity, Figure \ref{Example_on_C4} illustrates the plan if $G$ is a four-cycle.

\subsubsection*{A plan of n+4 moves to swap pebbles $\implies$ Hamiltonian Cycle}
Suppose a plan swapped the pebbles and used $N$ total moves. We will denote a hole to mean a vertex without a pebble. At any point in the plan, there are two holes, and they originated from $x$ and $y$. We partition the moves into whether the hole that shifted at that move originated from $x$ or $y$. Suppose that $j$ and $k$ are the sizes of each partition set ($j + k = N$). We define the hole trajectories as:

\noindent
$\text{traj}(x) := x_0 x_1, \cdots x_j$ \\
$\text{traj}(y) := y_0 y_1, \cdots y_k$, where $x_i, y_i \in V$

\begin{lem} (Three parts):
    \begin{enumerate}
        \item $x_0, x_j, y_0, y_k \in \{x,y\}$.
        \item $j,k > 0$.
        \item The union of trajectories $x$ and $y$ visits every vertex in $V$: $v_0, v_1, \cdots, v_{n-1}$.
    \end{enumerate}
\end{lem} \label{hole_trajectories}

    We first prove the lemma. 
        \begin{enumerate}
        \item Every $v \in V$ starts with a pebble, so the holes must start and return to $x$ and $y$.
        \item Without loss of generality, suppose a plan didn't use $y$. Then the pebble at $v_0$ can only shift between $x$ and $v_0$. Both of these vertices aren't in $V^-$, which is where it needs to be after the plan.
        \item For each $v \in V$, a different pebble starts and ends there. Right before the final vertex moves to $v$, there must have been a hole there.
        \end{enumerate}
Now, by the first two parts of lemma, trajectories $x$ and $y$ take the form: $\text{traj}(x) = xv_0 P_x u$ and $\text{traj}(y) = yv_0 P_y v$ where $u,v \in \{x,y\}$, and $P_x$ and $P_y$ are paths of even length that end with $v_0$ (they could be empty). If there is a plan of length $n+4$ that swapped the pebbles, then we will show that $P_x$ or $P_y$ is empty. Suppose that this is not the case. Then the trajectories take the form: $\text{traj}(x) = xv_0 P'_x v_0 u$ and $\text{traj}(y) = yv_0 P'_y v_0 v$ where $P'_x v_0 = P_x$ and $P'_y v_0 = P_y$. By the third part of lemma ~\ref{hole_trajectories}, $P'_x \cup P'_y = \{v_1, \cdots v_{n-1}\}$, so $|P'_x| + |P'_y| \geq n-1$. The length of the plan is $j + k = |\text{traj}(x)| + |\text{traj}(x)| - 2$.
Note that $|\text{traj}(x)| + |\text{traj}(y)| = 8 + |P'_x| + |P'_y|$ which includes the four total visits of $v_0$, and one visit each of $x$, $y$, $u$ and $v$. Thus, $|P'_x| + |P'_y| \geq n-1$ implies that the sum of the lengths of the trajectories is greater than $n+4$. This is a contradiction, so $P_x$ or $P_y$ is empty. 

We can assume without loss of generality that $P_y$ is empty, and $P_x$ visits $v_1, \cdots, v_{n-1}$. Each vertex appears exactly once in $P_x$ because if a vertex occured more than once, $P_x$ would be longer than $n-1$, contradicting the fact that the plan contains $n+4$ moves. Therefore, each vertex occurs exactly once and $v_{n-1}$ is connected to $v_0$, so we have a Hamiltonian cycle. There is a small technicality when $n=2$, but we can assume that all input graphs have at least four vertices.
\end{proof}

\begin{table}[h!]
\centering

\begin{tabular}{c|l|l}
\textbf{move \#} & \textbf{move} & \textbf{Configuration}
\\ \hline

0 & Initial 
& $\langle 0, 0, 1, -1, 1, \cdots, -1 \rangle$ 
\\

1 & ($v_0 \rightarrow x$) 
& $\langle 1, 0, 0, -1, 1, \cdots, -1 \rangle$ 
\\

2 ... n &  
$(v_1\!\to\! v_0)\ (v_2\!\to\! v_1),\ \dots,\ (v_{n-1}\!\to\! v_{n-2})$
& $\langle 1, 0, -1, 1, -1, \cdots, 0 \rangle$ 
\\

n+1 & ($v_0 \rightarrow y$)
& $\langle 1, -1, 0, 1, -1, \cdots, 0 \rangle$ 
\\

n+2 & ($x \rightarrow v_0$)
& $\langle 0, -1, 1, 1, -1, \cdots, 0 \rangle$ 
\\

n+3 & ($v_0 \rightarrow v_{n-1}$)
& $\langle 0, -1, 0, 1, -1, \dots, 1 \rangle$ 
\\

n+4 & ($y \rightarrow v_0$)
& $\langle 0, 0, -1, 1, -1, \cdots, 1 \rangle$ 
\\

\end{tabular}
\caption{The configurations are $\langle \C(x), \C(y), \C(v_0), \cdots, \C(v_{n-1}) \rangle$. The white and black pebbles have swapped because the starting and ending configurations are opposite.}
\label{Hamiltonian_cycle_moves}

\end{table}

\subsection*{Remark:}
There are certainly input graphs that do not have Hamiltonian Cycles, and it is possible to swap the black and white pebbles on the transformed graph using more than $n+4$ moves. Figure \ref{finite_plan_counterexample} shows an example.

\section*{Stepping into the Chess World}
We will adapt this reduction to Knight's Tour and Knight Exchange. Any knight's graph is bipartite, but the reduction does not immediately carry over because the geometry of chessboards restricts which vertex we can connect $x$ and $y$ to without connecting them to other vertices. For example, figure \ref{fig:not_extendable_example} shows a board where the $x,y$ extension is not possible. However, we will show that the closed knight's tour problem is NP-hard for a class of chessboards where the $x,y$ extension is possible. The details are shown in figure ~\ref{fig:reduction_is_extendable}. 

We now extend the result of McGown et. al \cite{McGownLeininger2002} by modifying their reduction. The Hamiltonian path and cycle problems when restricted to grid graphs are known to be NP-hard as shown by Itai et. al \cite{ItaiPapadimitriouSzwarcfiter1982} and Demaine and Rudoy \cite{DemaineRudoy2017}. We will reduce to their chess counterparts. In fact, we will employ the same reduction to show the hardness of both versions:

\begin{thm}
    The Open and Closed Knight's Tour Problems on connected chessboards is NP-hard.
    \label{thm:hardness_of_open_and_closed_knight_tour}
\end{thm}

\begin{proof}

Given an input grid graph, let $v \rightarrow M_v$ be a mapping from vertices to chessboards. Since the grid graph is bipartite, there is a natural partition $V = (A,B)$ based on the parity of the coordinates of vertices. 

\begin{itemize}
    \item For $v \in A$, let $b(v) = b_1 b_2 b_3 b_4$ be a binary string of indicator variables in the up, right, down, and left directions respectively, where $b_i$ is 1 if there is an edge in that direction or whether the vertex is located on that edge of the grid. Define $M_v$ to be the chessboard that $b(v)$ is mapped to as shown in tables ~\ref{tab:bit_strings_to_boards1} and ~\ref{tab:bit_strings_to_boards2}.
    \item For $v \in B$, permute $b(v) = b_1 b_2 b_3 b_4$ to be $b_2 b_1 b_4 b_3$. Let $Q$ be the chessboard that $b_2 b_1 b_4 b_3$ is mapped to as shown in tables ~\ref{tab:bit_strings_to_boards1} and ~\ref{tab:bit_strings_to_boards2}. Rotate $Q$ by $\pi$ and then apply the transpose. Define $M_v$ to be the resulting board.
\end{itemize}

For all $v \in V$, we will place $M_v$ in a larger conglomerate chessboard $M$. If $x$ connects up to $y$ in the grid graph, $M_y$ is situated 9 squares above $M_x$ in $M$. Perform the similar procedure for horizontal edges. By construction, $x$ and $y$ share an edge if and only if a rook can move between $M_x$ and $M_y$ without traveling through other boards. Therefore, $M$ is a connected chessboard if and only if the grid graph is connected, and we can restrict the inputs to be connected grid graphs. The number of squares in $M$ is linear with the number of vertices on the grid graph, and checking for the position of a vertex and determining its edges can be done in poly-time. Therefore, the reduction is poly-time.
For a concrete example of a how the reduction works, see figure ~\ref{reduction_example}.

The analysis for both versions will be similar to Mcgown's in their paper \cite{McGownLeininger2002}. We will omit the open version because it is nearly identical and just needs to be adapted for this new construction.

\subsubsection*{Hamiltonian Cycle in grid graph $\implies $ Closed knight's Tour.}

Suppose that there is a Hamiltonian cycle in the grid graph. We want to show that there is a closed knight's tour in $M$. The idea is to combine open knight's tours on $M_v$ to form a closed knight's tour on $M$. For any edge $xy$, there is a unique pair of squares in $M_x$ and $M_y$ that are adjacent in the knight's graph. We call them the transition squares of $M_x$ and $M_y$. Now, for any $v \in V$ with $t$ edges, there are ${t \choose 2}$ ways that the vertices in the cycle adjacent to $x$ can be positioned. Each way corresponds to a unique selection of two transition squares. Figures ~\ref{4_connected_cases}, ~\ref{3_connected_cases}, ~\ref{2_connected_cases}, and ~\ref{1-connected_cases} show that for any type of board, every selection of two transition squares has an open knight's tour that starts and ends there. Thus, we can combine knight's tours on each $M_v$ in the order that they appear in the cycle to build the closed knight's tour on $M$.

\subsubsection*{Closed Knight's Tour $\implies $ Hamiltonian Cycle in grid graph}

Suppose that $T$ is a closed knight's tour in $M$. Let $g: T \rightarrow V$ be defined by $g(s) = v$ when $s$ belongs to $M_v$. Consider the sequence $g(T)$ (apply $g$ to every element in $T$ pointwise). Every $v \in V$ appears in $g(T)$ since $T$ visits every square in $M$. 

We will now argue that vertices show up in $g(T)$ in blocks. For every $M_v$, the difference in light and dark squares is 1. All transition squares of $M_v$ are the same color, and a knight alternates between light and dark squares after each move. Therefore, the knight can only enter and leave $M_v$ once. Otherwise, the number of light and dark squares in $M_v$ that appear in $T$ would not differ by 1. We enumerate all the blocks to build a Hamiltonian Cycle in the grid graph.
\end{proof}

\begin{thm}
     The Knight Exchange Problem on connected chessboards is NP-hard
     \label{thm:hardness_of_knight_exchange}
\end{thm}

\begin{proof}
We will adapt the Pebble Swap reduction. Let us use the terminology that a connected chessboard is nicely extendable if we can make the $x,y$ extension as described in the pebble swap reduction while keeping the resulting chessboard still connected. Note that additional squares may also be added as long as they are isolated vertices in the knight's graph, since we can keep them as holes, and no pebbles can move there. All the chessboards in the reduction are nicely extendable due to the shape of $M_v$ for the vertices $v$ located on the right edge of the grid graph as shown in ~\ref{fig:reduction_is_extendable}. Since the Closed Knight's Tour problem on nicely extendable connected chessboards is NP-hard, we can apply the Pebble Swap reduction to show that the Knight Exchange problem is NP-hard on connected chessboards that are nicely extendable. The set of all connected chessboards is a superset of those that are nicely extendable, so Knight Exchange is NP hard.

\end{proof}

\begin{cor}
    Both Knight's Tour and Knight Exchange are NP-complete because they are each NP.
\end{cor}

\clearpage
\section*{Appendix}
    
\begin{figure}[h]
\begin{tikzpicture}[scale=1.2,
  vertex/.style={circle, draw, minimum size=12pt, inner sep=0pt},
  pebwhite/.style={circle, draw, fill=white, draw, inner sep=2pt},
  pebblack/.style={circle, draw, fill=black, inner sep=2pt},
  every label/.style={font=\small, inner sep=2pt} % <-- moves labels slightly away
]

% ====== VERTICES ======
\node[vertex,label=right:{$0$}] (0) at (0,0) {};

% left diamond
\node[vertex,label=above:{$4$}] (4) at (-1.2,1) {};
\node[vertex,label=right:{$5$}] (5) at (-2.2,0) {};
\node[vertex,label=below:{$6$}] (6) at (-1.2,-1) {};

% right diamond
\node[vertex,label=above:{$1$}] (1) at (1.2,1) {};
\node[vertex,label=right:{$2$}] (2) at (2.2,0) {};
\node[vertex,label=below:{$3$}] (3) at (1.2,-1) {};

% far right vertex
\node[vertex,label=right:{$8$}] (8) at (3.2,0) {};

% dangling x,y
\node[vertex,label=above:{$y$}] (y) at (0,1.6) {};
\node[vertex,label=below:{$x$}] (x) at (0,-1.6) {};

% ====== EDGES ======
\draw (4) -- (5) -- (6) -- (0) -- (4);
\draw (1) -- (2) -- (3) -- (0) -- (1);
\draw (1) -- (8) -- (3);
\draw (0) -- (y);
\draw (0) -- (x);

% ====== MINI PEBBLES ======
% white pebbles first
\foreach \v in {0,4,5,6,1,2,3,8} \node[pebwhite] at (\v) {};
% black pebbles on selected vertices
\foreach \v in {4,6,1,3} \node[pebblack] at (\v) {};

\end{tikzpicture}
\caption{One can check that it is possible to swap the black and white pebbles on $G'$, but there is no Hamiltonian Cycle in $G$. Note that $\text{traj}(x)$ and $\text{traj}(y)$ are longer than 2 for any plan that swaps the colors, so the length of the plan must be longer than 12 moves.}
\label{finite_plan_counterexample}
\end{figure}
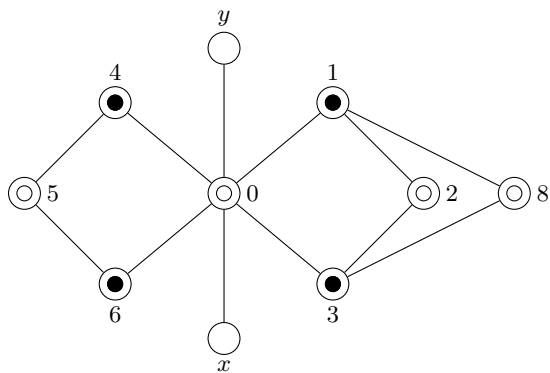
    % Macro: first arg = comma list of white-pebble vertices,
%        second arg = comma list of black-pebble vertices.
%        third arg = move (text)
\newcommand{\PebbleState}[3]{%
  \begin{tikzpicture}[scale=2, every node/.style={font=\small}, baseline={(0.base)}]
    % vertex positions (defined once)
    \node (0) at (0.5,1)  [circle, draw, minimum size=10pt] {};
    \node (1) at (1,0.5)  [circle, draw, minimum size=10pt] {};
    \node (2) at (0.5,0)  [circle, draw, minimum size=10pt] {};
    \node (3) at (0,0.5)  [circle, draw, minimum size=10pt] {};
    \node (x) at (0,1.5)  [circle, draw, minimum size=10pt] {};
    \node (y) at (1,1.5)  [circle, draw, minimum size=10pt] {};

    % --- label above ---
    \node at (0.5,1.9) {\scriptsize #3};

    % edges
    \foreach \a/\b in {0/1,1/2,2/3,3/0,0/x,0/y} \draw (\a) -- (\b);

    % white pebbles (draw first so black can override)
    \foreach \v in {#1} \node at (\v) [circle, fill=white, draw, inner sep=2pt] {};

    % black pebbles (draw after whites to override)
    \foreach \v in {#2} \node at (\v) [circle, fill=black, draw, inner sep=2pt] {};

    % labels (adjust offset or size if desired)
    \foreach \v/\lab in {0/0,1/1,2/2,3/3,x/x,y/y}
      \node at ($(\v)+(0,0.18)$) {\scriptsize $\lab$};
  \end{tikzpicture}%
}

\begin{figure}[h!]
\begin{tabular}{l l}
\textbf{Move \#} & \textbf{Picture} \\[0.2cm] \hline \\

\textbf{0 through 1:} &
\PebbleState{0,2}{1,3}{Initial} \hspace{1cm}
\PebbleState{x,2}{1,3}{$0 \rightarrow x$}
\\ [2cm] \\

\textbf{2 through n:} &
\PebbleState{x,2}{0,3}{$1 \rightarrow 0$}
\hspace{1cm}
\PebbleState{x,1}{0,3}{$2 \rightarrow 1$}
\hspace{1cm}
\PebbleState{x,1}{0,2}{$3 \rightarrow 2$}
\\[1.5cm]\\

\textbf{n+1 through n+4:} &
\PebbleState{x,1}{y,2}{$0 \rightarrow y$}
\hspace{1cm}
\PebbleState{1,0}{y,2}{$x \rightarrow 0$}
\hspace{1cm}
\PebbleState{1,3}{y,2}{$0 \rightarrow 3$}
\hspace{1cm}
\PebbleState{1,3}{0,2}{$y \rightarrow 0$}
\end{tabular}

\caption{$G$ is a four-cycle, and $G'$ contains two additional vertices $x$ and $y$ attached to $0$ as shown. In this example, $n=4$.}
\label{Example_on_C4}
\end{figure}

    \begin{table}[h!]
\centering
\scriptsize
\renewcommand{\arraystretch}{1.2}

% ====================================================
% LEFT HALF (first 8 bitstrings)
% ====================================================
\begin{minipage}{0.495\textwidth}
\centering
\resizebox{\textwidth}{!}{%
\begin{tabular}{c|c}
\textbf{Edges} & \textbf{Chessboard} \\
\hline

\texttt{(0,0,0,0)} &
$\begin{pmatrix}
0&0&0&0&0&0&0&0&0\\
0&0&0&0&0&0&0&0&0\\
0&0&0&0&0&0&0&0&0\\
0&0&0&0&0&0&0&0&0\\
0&0&0&0&0&0&0&0&0\\
0&0&0&0&0&0&0&0&0\\
0&0&0&0&0&0&0&0&0\\
0&0&0&0&0&0&0&0&0\\
0&0&0&0&0&0&0&0&0
\end{pmatrix}$ \\ \hline

\texttt{(0,0,0,1)} &
$\begin{pmatrix}
0&0&0&0&0&0&0&0&0\\
1&1&1&1&1&1&0&0&0\\
0&1&1&1&1&1&1&1&0\\
0&1&1&1&1&1&1&1&0\\
0&1&1&1&1&1&1&1&0\\
0&1&1&1&1&1&1&1&0\\
0&1&1&1&1&1&1&0&0\\
0&0&1&1&1&1&1&0&0\\
0&0&0&0&0&0&0&0&0
\end{pmatrix}$ \\ \hline

\texttt{(0,0,1,0)} &
$\begin{pmatrix}
0&0&0&0&0&0&0&0&0\\
0&0&1&1&1&1&0&0&0\\
0&1&1&1&1&1&1&1&0\\
0&1&1&1&1&1&1&1&0\\
0&1&1&1&1&1&1&1&0\\
0&1&1&1&1&1&1&1&0\\
0&1&1&1&1&1&1&0&0\\
0&1&1&1&1&1&1&0&0\\
0&1&0&0&0&0&0&0&0
\end{pmatrix}$ \\ \hline

\texttt{(0,0,1,1)} &
$\begin{pmatrix}
0&0&0&0&0&0&0&0&0\\
1&1&1&1&1&1&0&0&0\\
0&1&1&1&1&1&1&1&0\\
0&1&1&1&1&1&1&1&0\\
0&1&1&1&1&1&1&1&0\\
0&1&1&1&1&1&1&1&0\\
0&1&1&1&1&1&1&0&0\\
0&1&1&1&1&1&1&0&0\\
0&1&0&0&0&0&0&0&0
\end{pmatrix}$ \\ \hline

\end{tabular}}
\end{minipage}
\hfill
% ====================================================
% RIGHT HALF (last 8 bitstrings)
% ====================================================
\begin{minipage}{0.495\textwidth}
\centering
\resizebox{\textwidth}{!}{%
\begin{tabular}{c|c}
\textbf{Edges} & \textbf{Chessboard} \\
\hline

\texttt{(1,0,0,0)} &
$\begin{pmatrix}
0&0&0&0&0&0&0&1&0\\
0&0&1&1&1&1&0&1&0\\
0&1&1&1&1&1&1&1&0\\
0&1&1&1&1&1&1&1&0\\
0&1&1&1&1&1&1&1&0\\
0&1&1&1&1&1&1&1&0\\
0&1&1&1&1&1&1&0&0\\
0&0&1&1&1&1&1&0&0\\
0&0&0&0&0&0&0&0&0
\end{pmatrix}$ \\ \hline

\texttt{(1,0,0,1)} &
$\begin{pmatrix}
0&0&0&0&0&0&0&1&0\\
1&1&1&1&1&1&0&1&0\\
0&1&1&1&1&1&1&1&0\\
0&1&1&1&1&1&1&1&0\\
0&1&1&1&1&1&1&1&0\\
0&1&1&1&1&1&1&1&0\\
0&1&1&1&1&1&1&0&0\\
0&0&1&1&1&1&1&0&0\\
0&0&0&0&0&0&0&0&0
\end{pmatrix}$ \\ \hline

\texttt{(1,0,1,0)} &
$\begin{pmatrix}
0&0&0&0&0&0&0&1&0\\
0&0&1&1&1&1&0&1&0\\
0&1&1&1&1&1&1&1&0\\
0&1&1&1&1&1&1&1&0\\
0&1&1&1&1&1&1&1&0\\
0&1&1&1&1&1&1&1&0\\
0&1&1&1&1&1&1&0&0\\
0&1&1&1&1&1&1&0&0\\
0&1&0&0&0&0&0&0&0
\end{pmatrix}$ \\ \hline

\texttt{(1,0,1,1)} &
$\begin{pmatrix}
0&0&0&0&0&0&0&1&0\\
1&1&1&1&1&1&0&1&0\\
0&1&1&1&1&1&1&1&0\\
0&1&1&1&1&1&1&1&0\\
0&1&1&1&1&1&1&1&0\\
0&1&1&1&1&1&1&1&0\\
0&1&1&1&1&1&1&0&0\\
0&1&1&1&1&1&1&0&0\\
0&1&0&0&0&0&0&0&0
\end{pmatrix}$ \\ \hline

\end{tabular}}
\end{minipage}

\caption{The mapping between edge strings and chessboards are shown above where the edge strings are in the order (up, right, down, left). For the edges, a 1 indicates the presence of an edge in that direction, and 0 indicates the absence of an edge. For the chessboards, a 1 indicates the presence of a square, and a 0 indicates the absence of a square. The (1,1,1,1) chessboard has four transition squares. To construct the other chessboards whenever an edge isn't present we delete the square sticking out and one adjacent square next to it. That way, the number of dark and light squares are conserved, and there is no longer a transition square for that edge. No Hamiltonian Cycle or Path will include an isolated vertex, so (0,0,0,0) is mapped to the empty chessboard. The mappings are continued in the next table.}
\label{tab:bit_strings_to_boards1}
\end{table}

    \begin{table}[h!]
\centering
\scriptsize
\renewcommand{\arraystretch}{1.2}

% ====================================================
% LEFT HALF (first 8 bitstrings)
% ====================================================
\begin{minipage}{0.495\textwidth}
\centering
\resizebox{\textwidth}{!}{%
\begin{tabular}{c|c}
\textbf{Edges} & \textbf{Chessboard} \\
\hline

\texttt{(0,1,0,0)} &
$\begin{pmatrix}
0&0&0&0&0&0&0&0&0\\
0&0&1&1&1&1&0&0&0\\
0&1&1&1&1&1&1&1&0\\
0&1&1&1&1&1&1&1&0\\
0&1&1&1&1&1&1&1&0\\
0&1&1&1&1&1&1&1&0\\
0&1&1&1&1&1&1&0&0\\
0&0&1&1&1&1&1&1&1\\
0&0&0&0&0&0&0&0&0
\end{pmatrix}$ \\ \hline

\texttt{(0,1,0,1)} &
$\begin{pmatrix}
0&0&0&0&0&0&0&0&0\\
1&1&1&1&1&1&0&0&0\\
0&1&1&1&1&1&1&1&0\\
0&1&1&1&1&1&1&1&0\\
0&1&1&1&1&1&1&1&0\\
0&1&1&1&1&1&1&1&0\\
0&1&1&1&1&1&1&0&0\\
0&0&1&1&1&1&1&1&1\\
0&0&0&0&0&0&0&0&0
\end{pmatrix}$ \\ \hline

\texttt{(0,1,1,0)} &
$\begin{pmatrix}
0&0&0&0&0&0&0&0&0\\
0&0&1&1&1&1&0&0&0\\
0&1&1&1&1&1&1&1&0\\
0&1&1&1&1&1&1&1&0\\
0&1&1&1&1&1&1&1&0\\
0&1&1&1&1&1&1&1&0\\
0&1&1&1&1&1&1&0&0\\
0&1&1&1&1&1&1&1&1\\
0&1&0&0&0&0&0&0&0
\end{pmatrix}$ \\ \hline

\texttt{(0,1,1,1)} &
$\begin{pmatrix}
0&0&0&0&0&0&0&0&0\\
1&1&1&1&1&1&0&0&0\\
0&1&1&1&1&1&1&1&0\\
0&1&1&1&1&1&1&1&0\\
0&1&1&1&1&1&1&1&0\\
0&1&1&1&1&1&1&1&0\\
0&1&1&1&1&1&1&0&0\\
0&1&1&1&1&1&1&1&1\\
0&1&0&0&0&0&0&0&0
\end{pmatrix}$ \\ \hline

\end{tabular}}
\end{minipage}
\hfill
% ====================================================
% RIGHT HALF (last 8 bitstrings)
% ====================================================
\begin{minipage}{0.495\textwidth}
\centering
\resizebox{\textwidth}{!}{%
\begin{tabular}{c|c}
\textbf{Edges} & \textbf{Chessboard} \\
\hline

\texttt{(1,1,0,0)} &
$\begin{pmatrix}
0&0&0&0&0&0&0&1&0\\
0&0&1&1&1&1&0&1&0\\
0&1&1&1&1&1&1&1&0\\
0&1&1&1&1&1&1&1&0\\
0&1&1&1&1&1&1&1&0\\
0&1&1&1&1&1&1&1&0\\
0&1&1&1&1&1&1&0&0\\
0&0&1&1&1&1&1&1&1\\
0&0&0&0&0&0&0&0&0
\end{pmatrix}$ \\ \hline

\texttt{(1,1,0,1)} &
$\begin{pmatrix}
0&0&0&0&0&0&0&1&0\\
1&1&1&1&1&1&0&1&0\\
0&1&1&1&1&1&1&1&0\\
0&1&1&1&1&1&1&1&0\\
0&1&1&1&1&1&1&1&0\\
0&1&1&1&1&1&1&1&0\\
0&1&1&1&1&1&1&0&0\\
0&0&1&1&1&1&1&1&1\\
0&0&0&0&0&0&0&0&0
\end{pmatrix}$ \\ \hline

\texttt{(1,1,1,0)} &
$\begin{pmatrix}
0&0&0&0&0&0&0&1&0\\
0&0&1&1&1&1&0&1&0\\
0&1&1&1&1&1&1&1&0\\
0&1&1&1&1&1&1&1&0\\
0&1&1&1&1&1&1&1&0\\
0&1&1&1&1&1&1&1&0\\
0&1&1&1&1&1&1&0&0\\
0&1&1&1&1&1&1&1&1\\
0&1&0&0&0&0&0&0&0
\end{pmatrix}$ \\ \hline

\texttt{(1,1,1,1)} &
$\begin{pmatrix}
0&0&0&0&0&0&0&1&0\\
1&1&1&1&1&1&0&1&0\\
0&1&1&1&1&1&1&1&0\\
0&1&1&1&1&1&1&1&0\\
0&1&1&1&1&1&1&1&0\\
0&1&1&1&1&1&1&1&0\\
0&1&1&1&1&1&1&0&0\\
0&1&1&1&1&1&1&1&1\\
0&1&0&0&0&0&0&0&0
\end{pmatrix}$ \\ \hline

\end{tabular}}
\end{minipage}

\caption{The rest of the mappings between edge strings and boards are shown above.}
\label{tab:bit_strings_to_boards2}
\end{table}

    \begin{figure}
\centering
\begin{tikzpicture}[scale=1.8, yscale=-1, every node/.style={circle, fill=black, inner sep=1.5pt}]

  % --- Vertices ---
  \node (node_0_0) at (0,0) {};
  \node (node_0_1) at (0,1) {};
  \node (node_0_2) at (0,2) {};
  \node (node_0_3) at (0,3) {};
  \node (node_0_5) at (0,5) {};
  \node (node_1_1) at (1,1) {};
  \node (node_1_2) at (1,2) {};
  \node (node_1_3) at (1,3) {};
  \node (node_1_4) at (1,4) {};
  \node (node_1_5) at (1,5) {};
  \node (node_2_0) at (2,0) {};
  \node (node_2_3) at (2,3) {};
  \node (node_2_4) at (2,4) {};
  \node (node_2_5) at (2,5) {};
  \node (node_3_0) at (3,0) {};
  \node (node_3_1) at (3,1) {};
  \node (node_3_2) at (3,2) {};
  \node (node_3_3) at (3,3) {};
  \node (node_3_4) at (3,4) {};
  \node (node_3_5) at (3,5) {};
  \node (node_4_0) at (4,0) {};
  \node (node_4_1) at (4,1) {};
  \node (node_4_2) at (4,2) {};
  \node (node_4_3) at (4,3) {};
  \node (node_4_4) at (4,4) {};
  \node (node_4_5) at (4,5) {};
  \node (node_5_0) at (5,0) {};
  \node (node_5_1) at (5,1) {};
  \node (node_5_2) at (5,2) {};
  \node (node_5_5) at (5,5) {};

  % --- Edges ---
  \draw (node_0_0) -- (node_0_1);
  \draw (node_0_1) -- (node_0_2);
  \draw (node_0_1) -- (node_1_1);
  \draw (node_0_2) -- (node_0_3);
  \draw (node_0_2) -- (node_1_2);
  \draw (node_0_3) -- (node_1_3);
  \draw (node_0_5) -- (node_1_5);
  \draw (node_1_1) -- (node_1_2);
  \draw (node_1_3) -- (node_1_4);
  \draw (node_1_4) -- (node_1_5);
  \draw (node_1_5) -- (node_2_5);
  \draw (node_2_0) -- (node_3_0);
  \draw (node_2_3) -- (node_2_4);
  \draw (node_2_3) -- (node_3_3);
  \draw (node_2_4) -- (node_2_5);
  \draw (node_2_4) -- (node_3_4);
  \draw (node_2_5) -- (node_3_5);
  \draw (node_3_0) -- (node_3_1);
  \draw (node_3_1) -- (node_3_2);
  \draw (node_3_1) -- (node_4_1);
  \draw (node_3_2) -- (node_3_3);
  \draw (node_3_2) -- (node_4_2);
  \draw (node_3_3) -- (node_3_4);
  \draw (node_3_3) -- (node_4_3);
  \draw (node_3_4) -- (node_3_5);
  \draw (node_3_4) -- (node_4_4);
  \draw (node_3_5) -- (node_4_5);
  \draw (node_4_0) -- (node_4_1);
  \draw (node_4_0) -- (node_5_0);
  \draw (node_4_1) -- (node_4_2);
  \draw (node_4_1) -- (node_5_1);
  \draw (node_4_2) -- (node_4_3);
  \draw (node_4_2) -- (node_5_2);
  \draw (node_4_3) -- (node_4_4);
  \draw (node_4_4) -- (node_4_5);
  \draw (node_4_5) -- (node_5_5);
  \draw (node_5_1) -- (node_5_2);

\end{tikzpicture}
\includegraphics[scale=0.52]{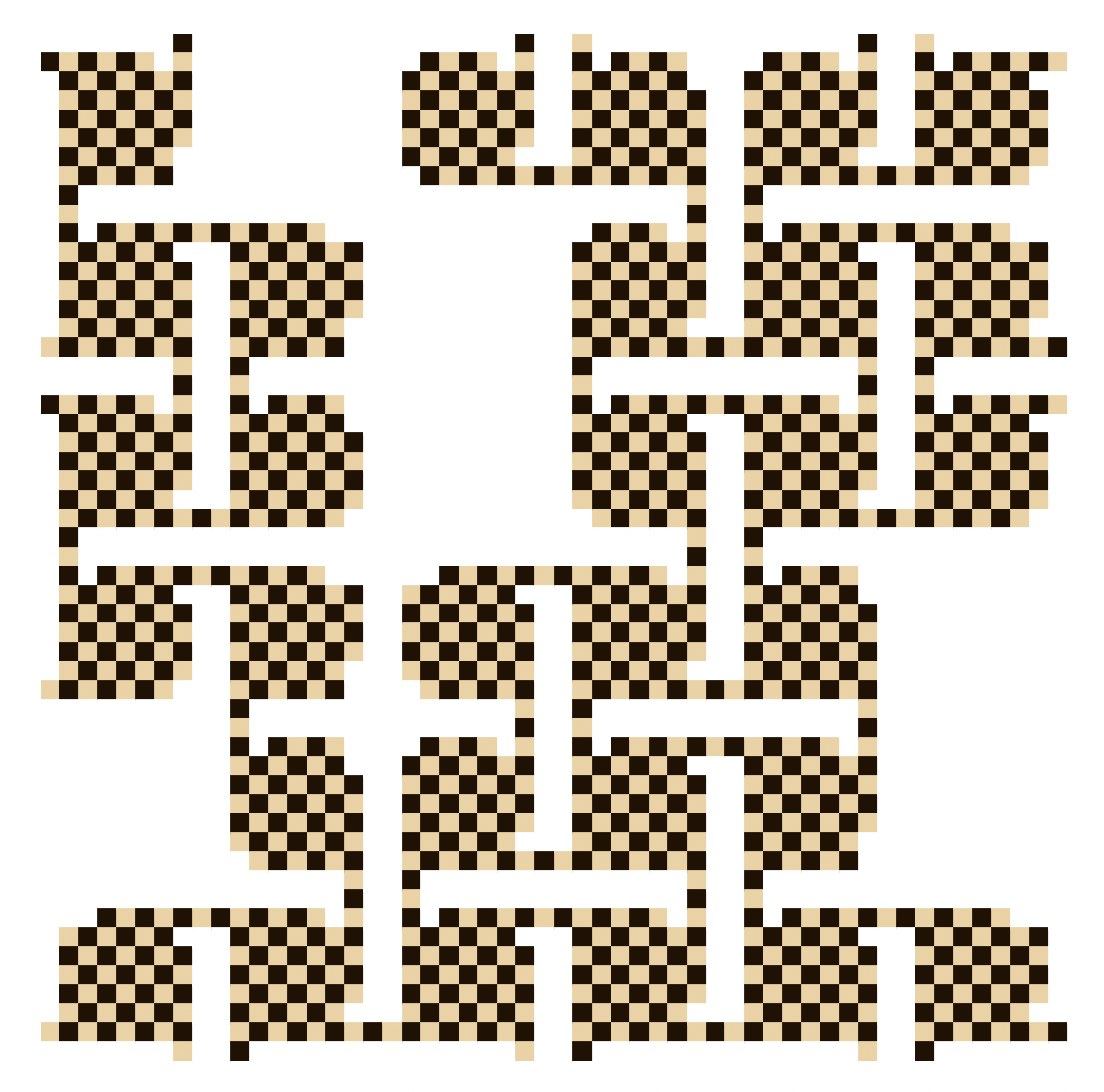}
\caption{A grid graph and the resulting construction is shown. Each vertex in the grid graph gets mapped to a 9 by 9 chessboard with holes. There is only one way to get between any two adjacent chessboards that are connected by an edge in the grid graph, and there is no way to get between them if they are not connected. Also notice that all the transition squares are the same color, and the number of squares of this color is the majority by 1 for each chessboard.}
\label{reduction_example}
\end{figure}
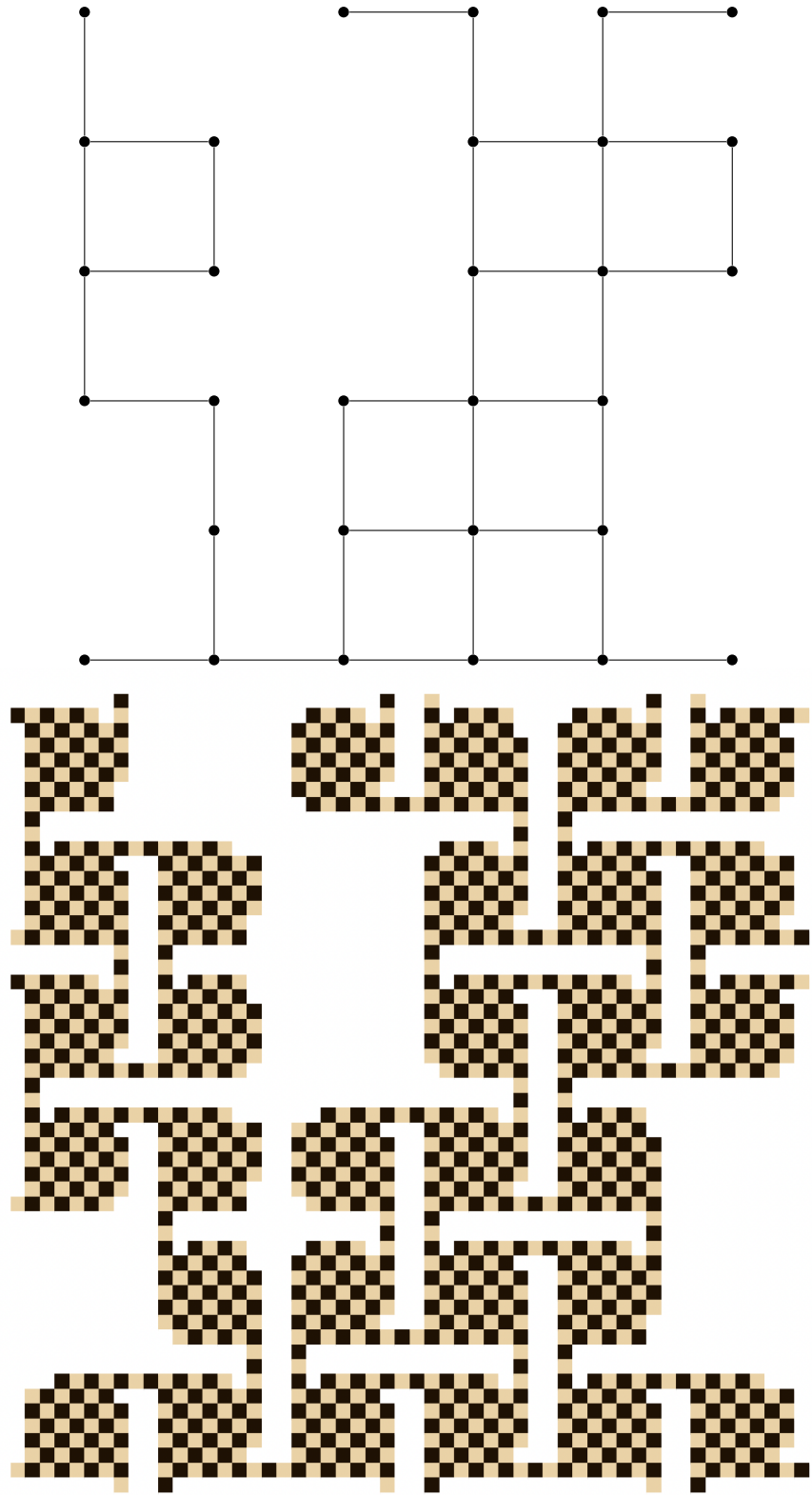
    \begin{figure}[h] \centering \includegraphics[scale=0.23]{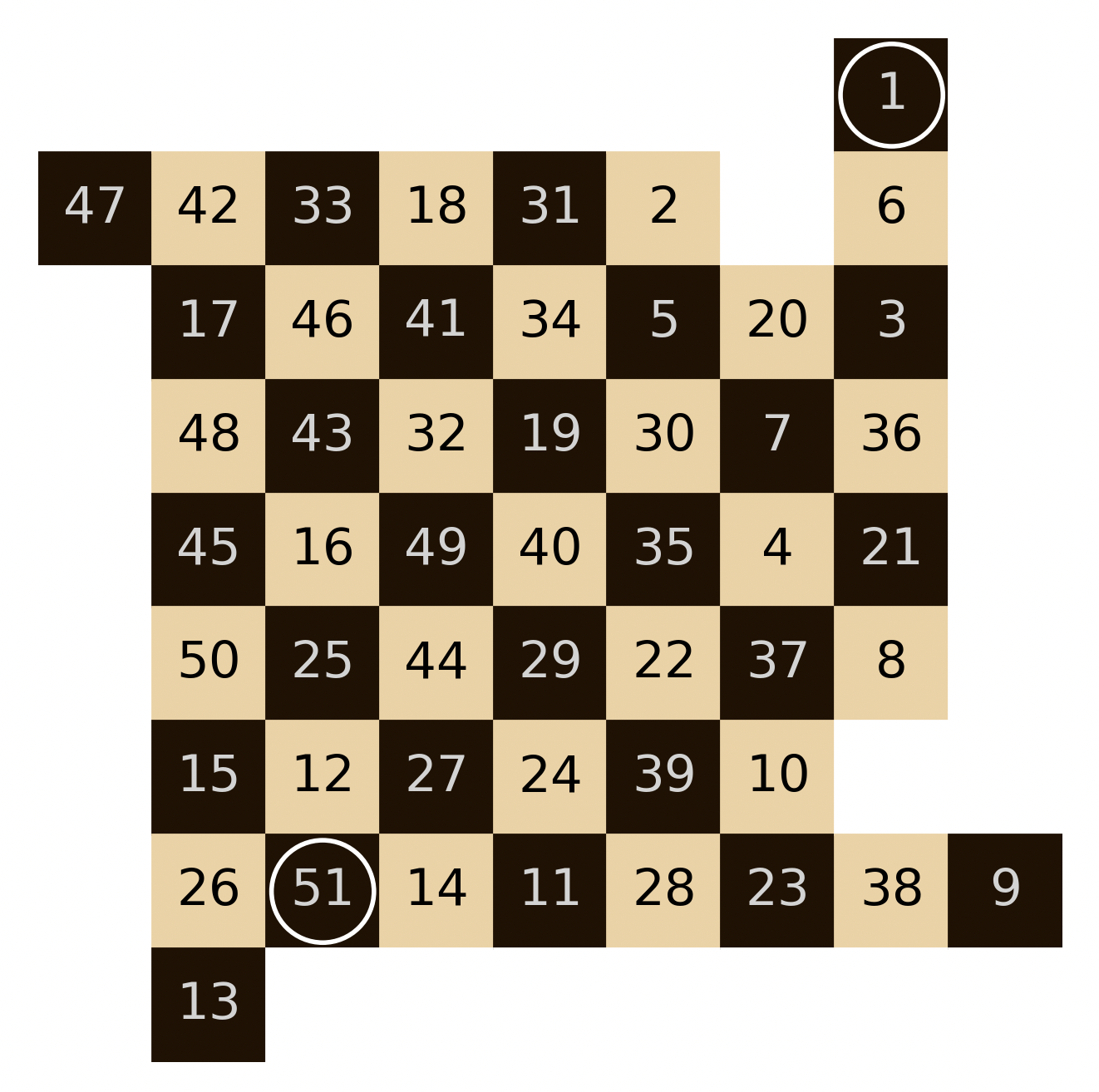} \includegraphics[scale=0.23]{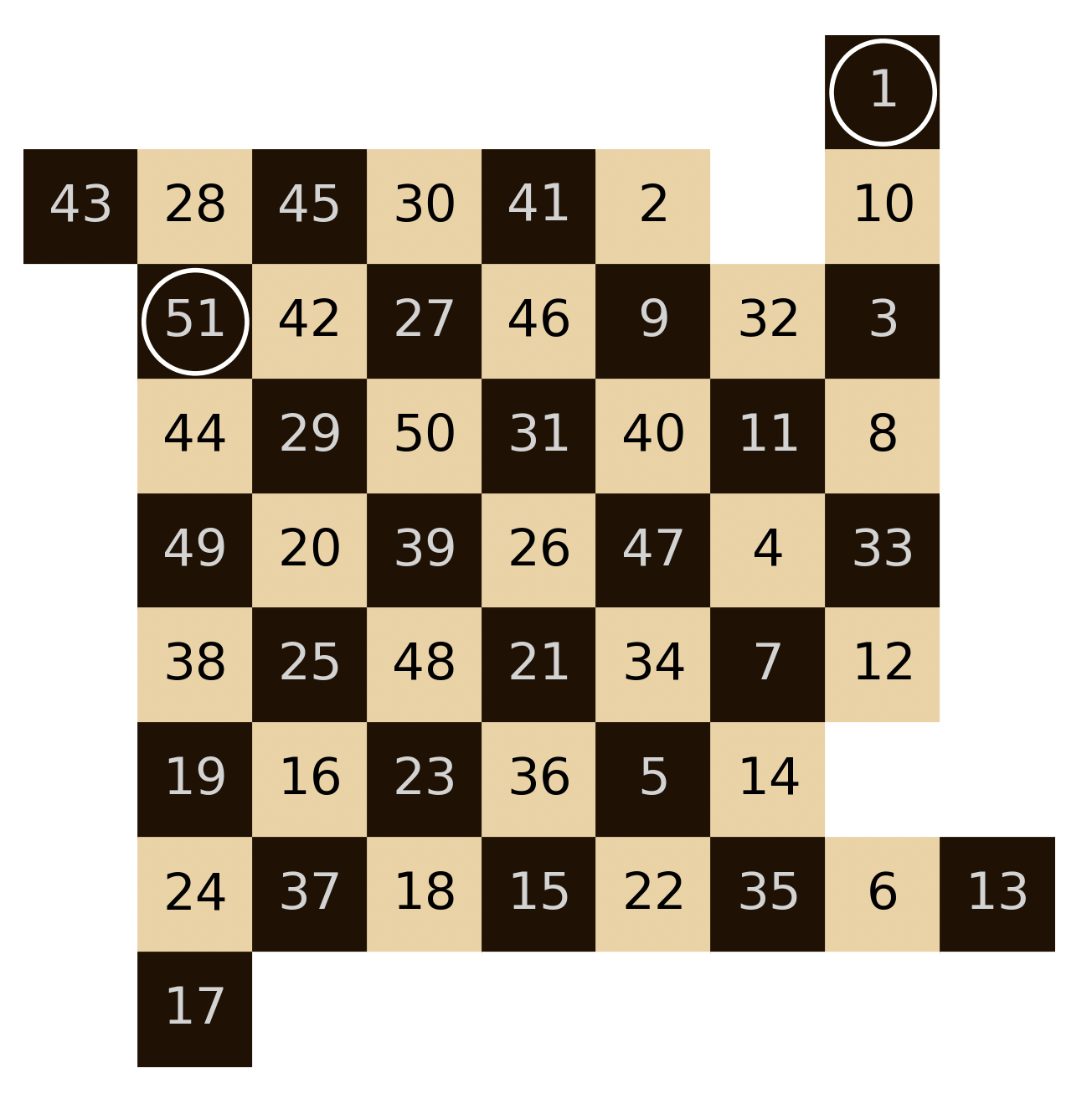} \includegraphics[scale=0.23]{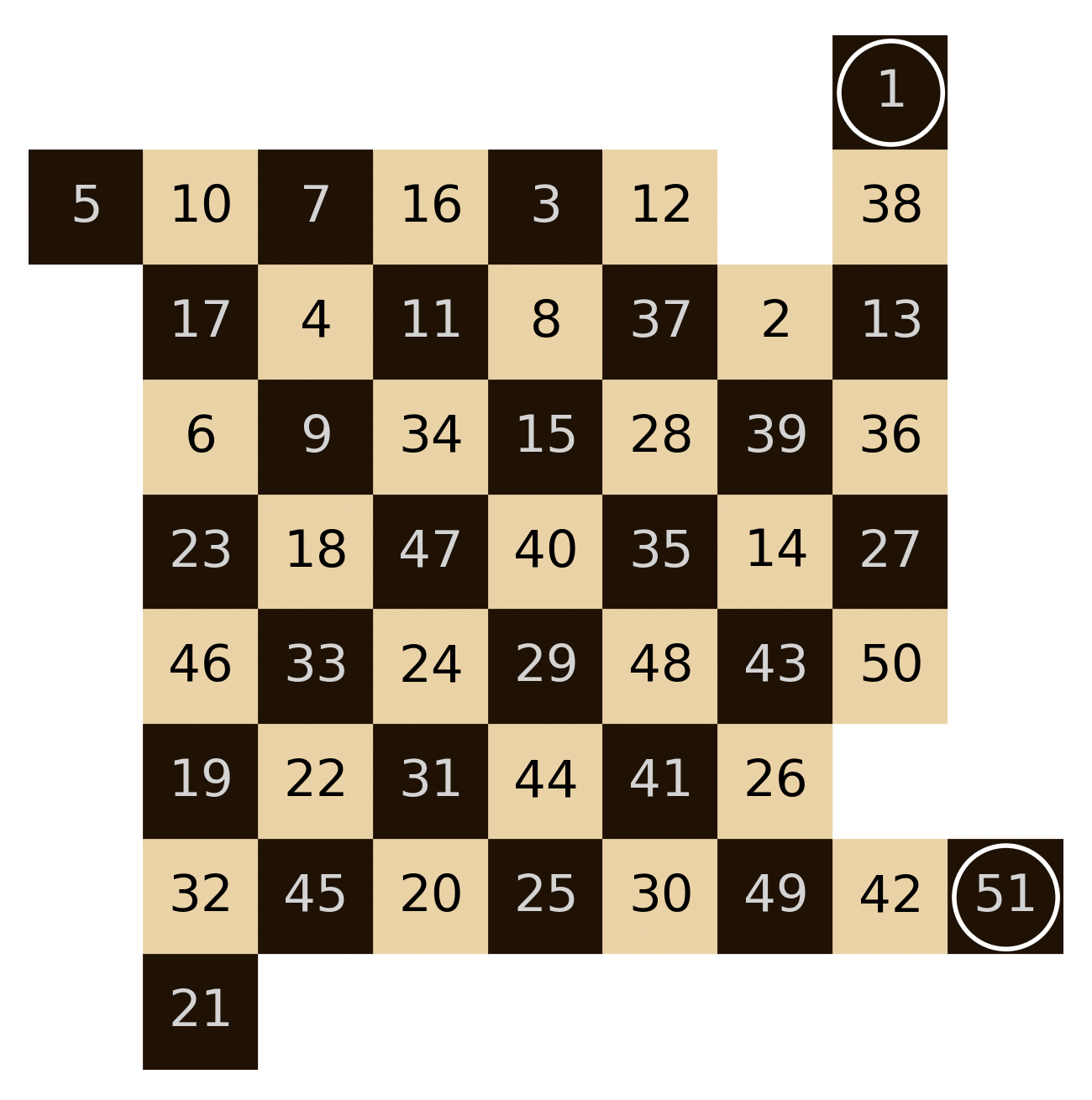} \caption{A board is $k-$connected if its vertex in the grid graph has $k$ edges. Transition squares are entry/leaving points to other chessboards. For any pair of transition squares on the unique 4-connected board up to symmetry, there is an open knight's tour. The visit times on each square is enumerated, and transition squares are circled.} \label{4_connected_cases} \includegraphics[scale=0.23]{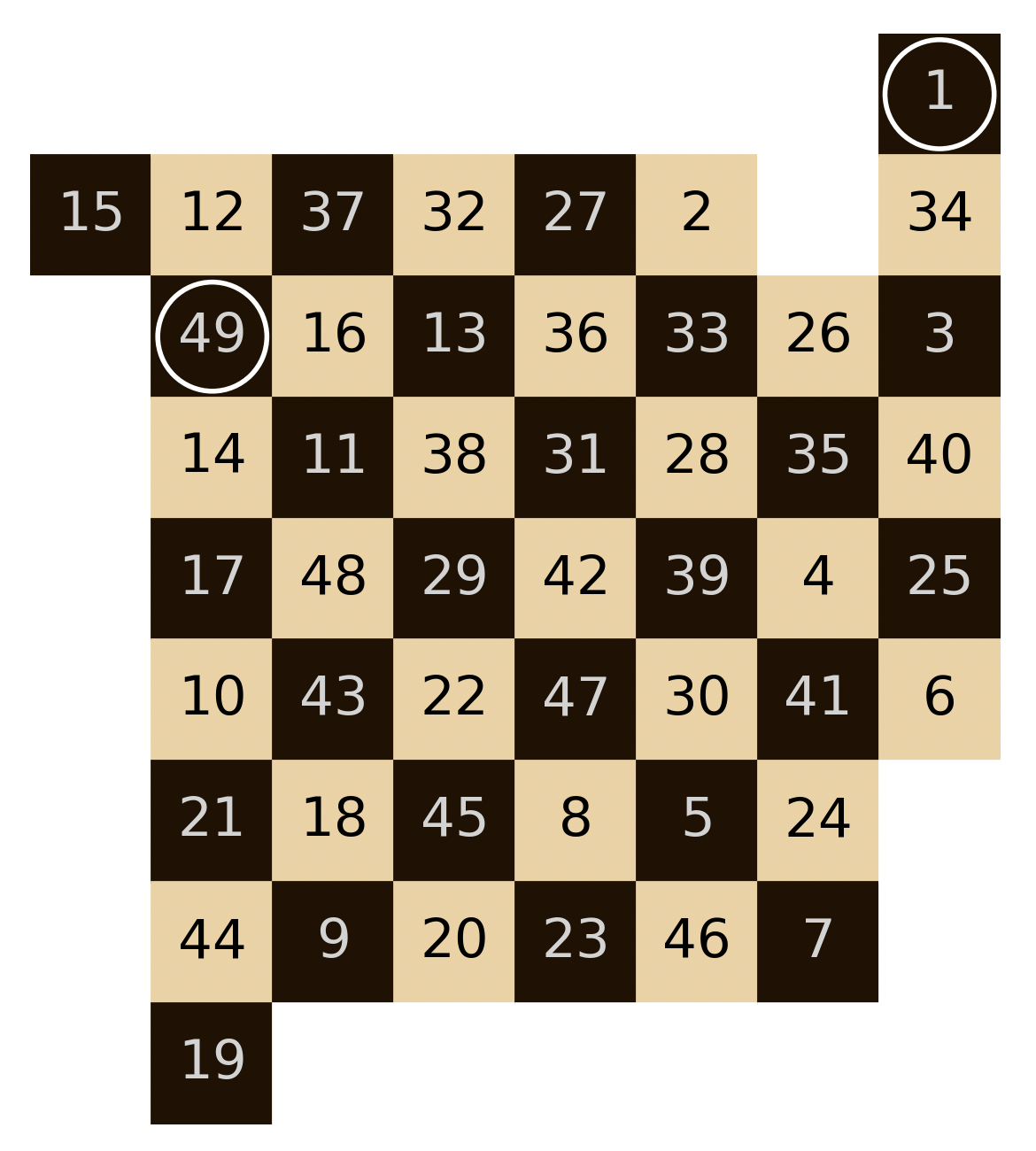} \includegraphics[scale=0.23]{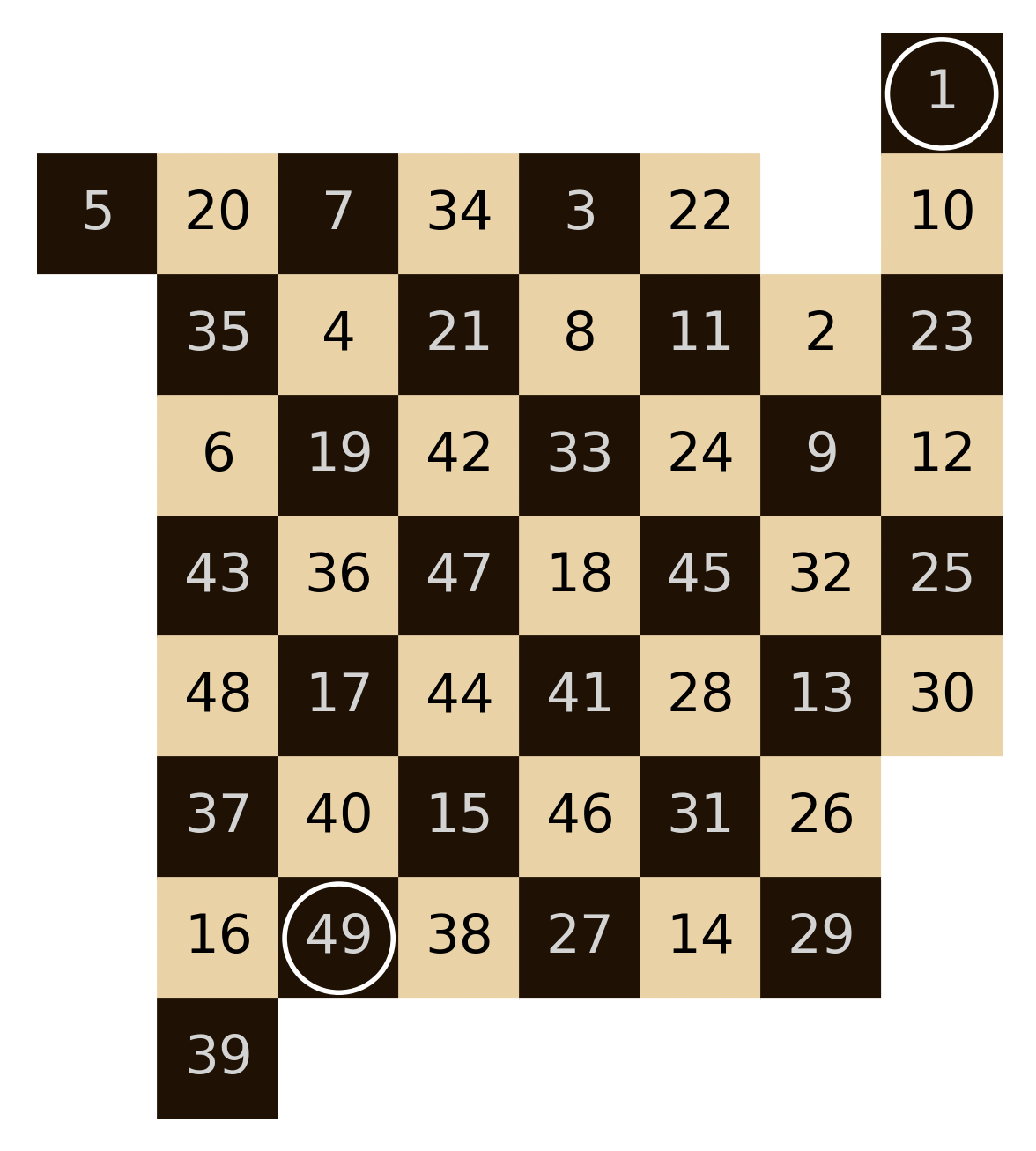} \includegraphics[scale=0.23]{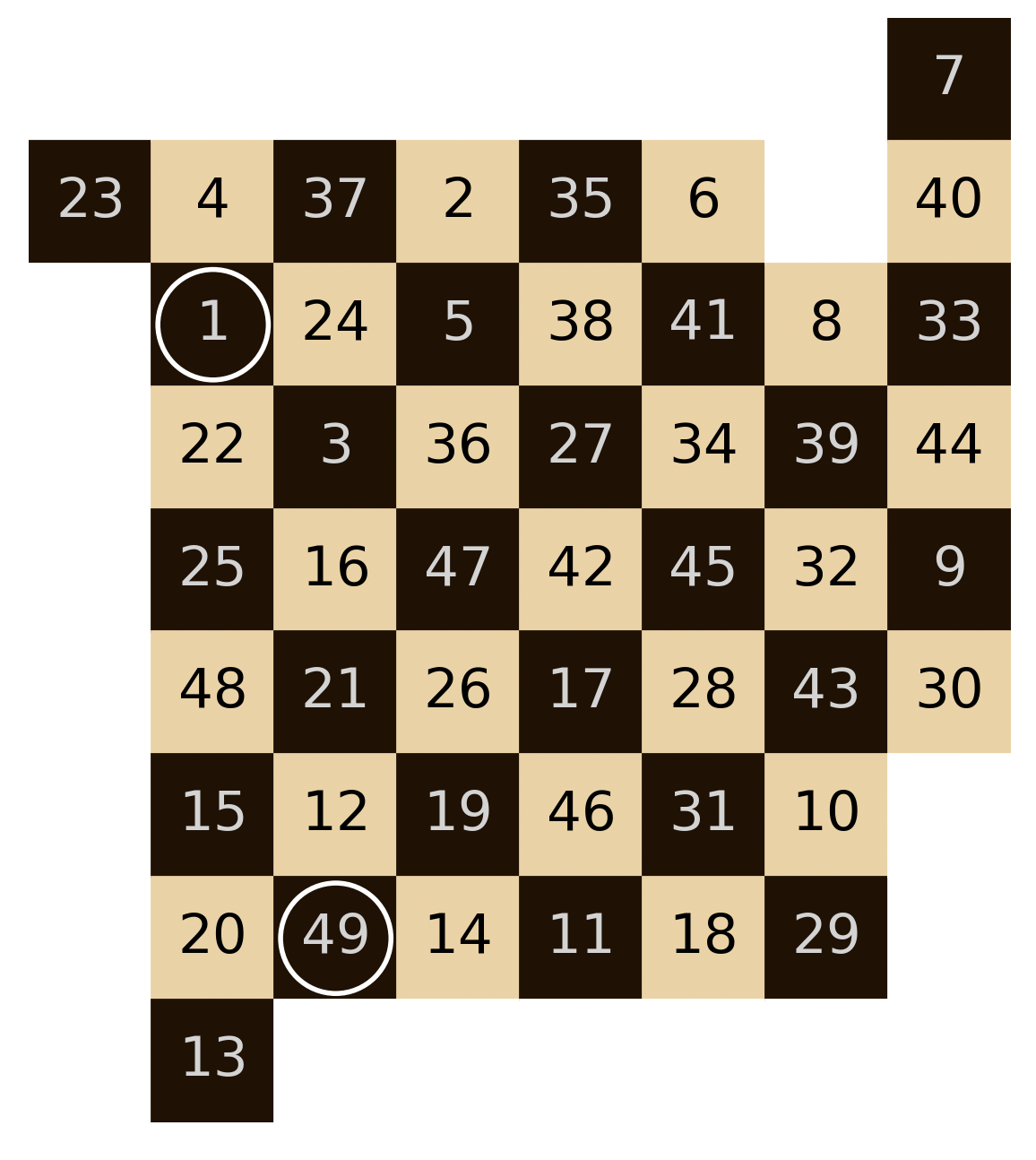} \includegraphics[scale=0.23]{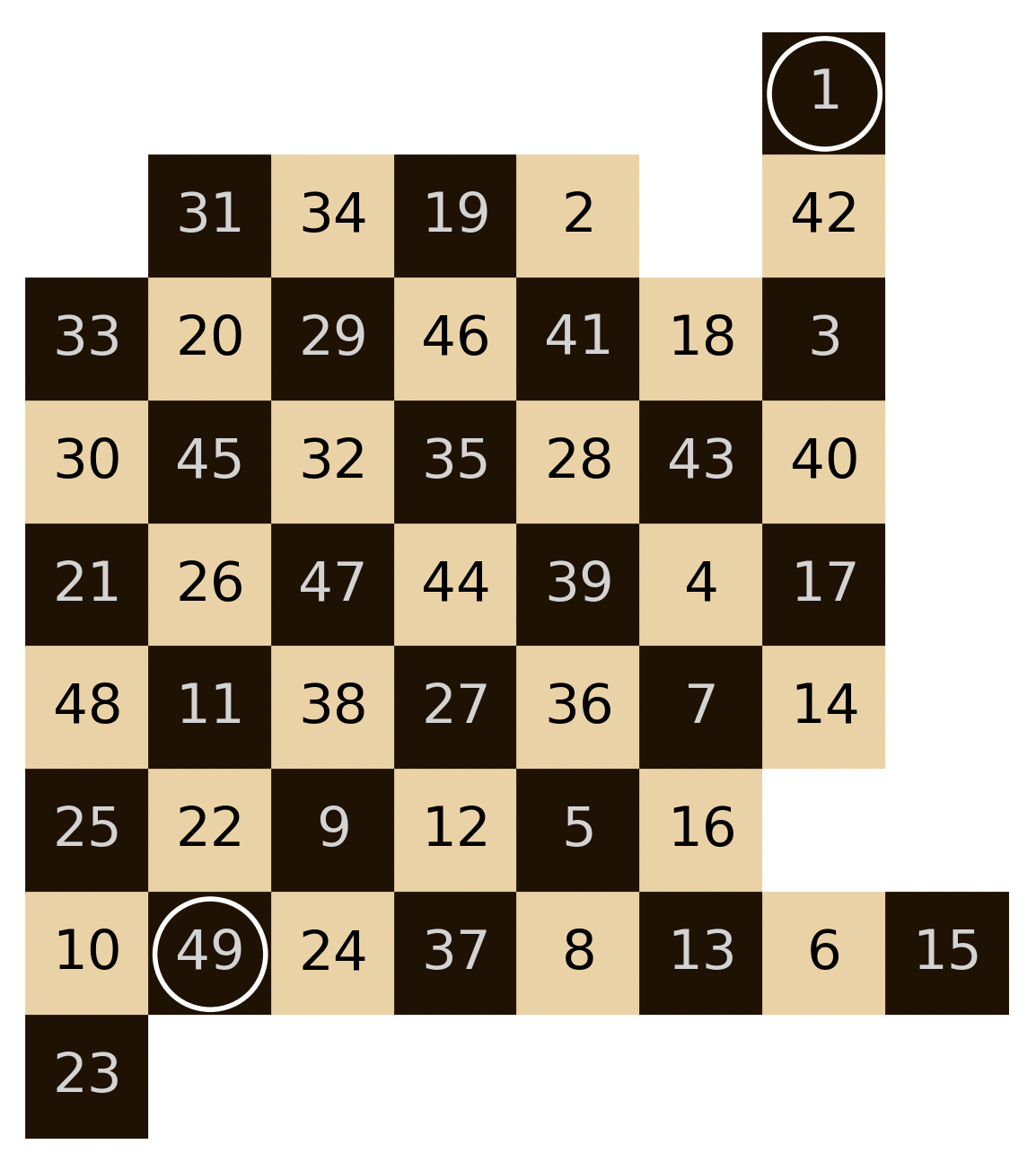} \includegraphics[scale=0.23]{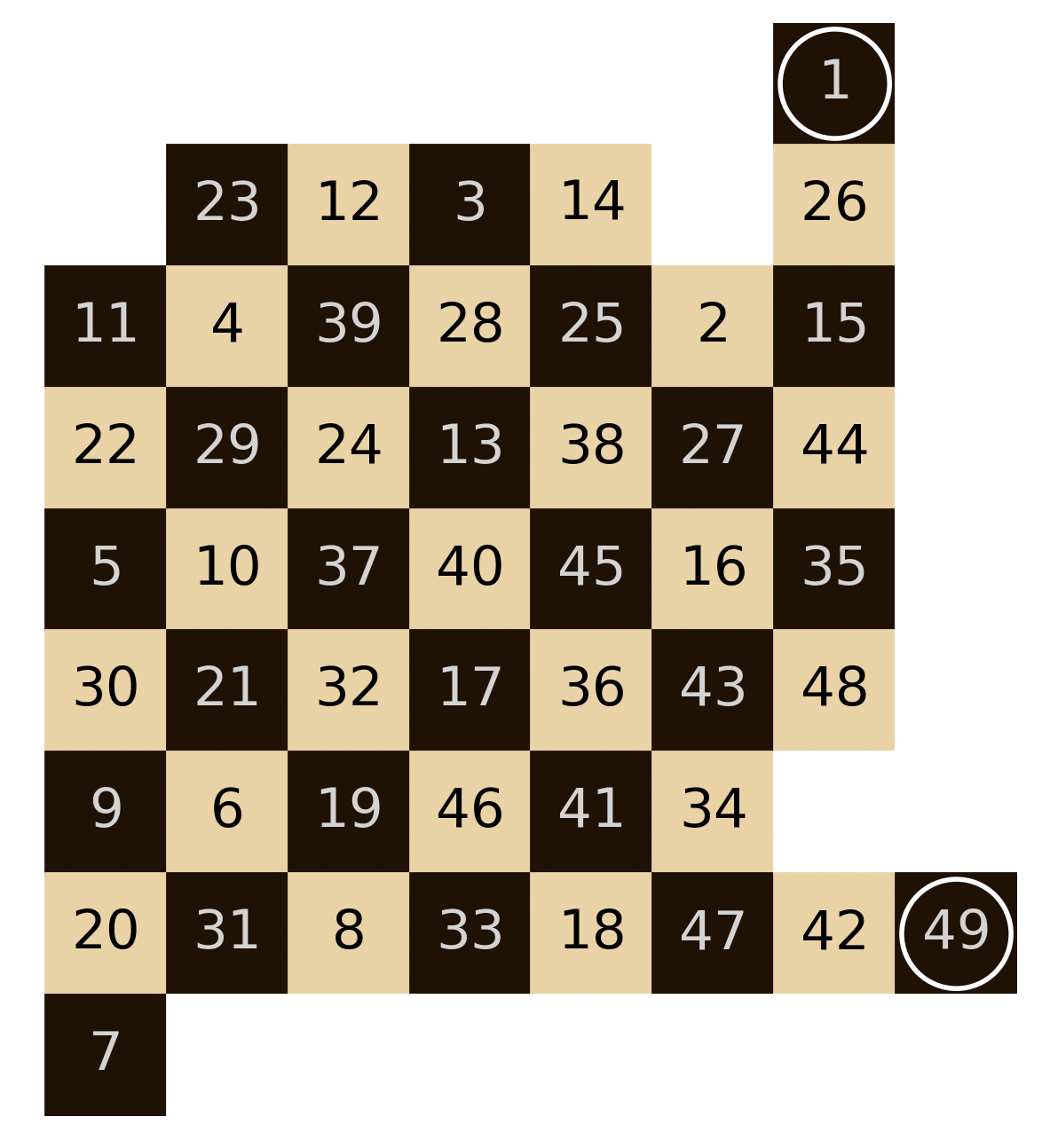} \includegraphics[scale=0.23]{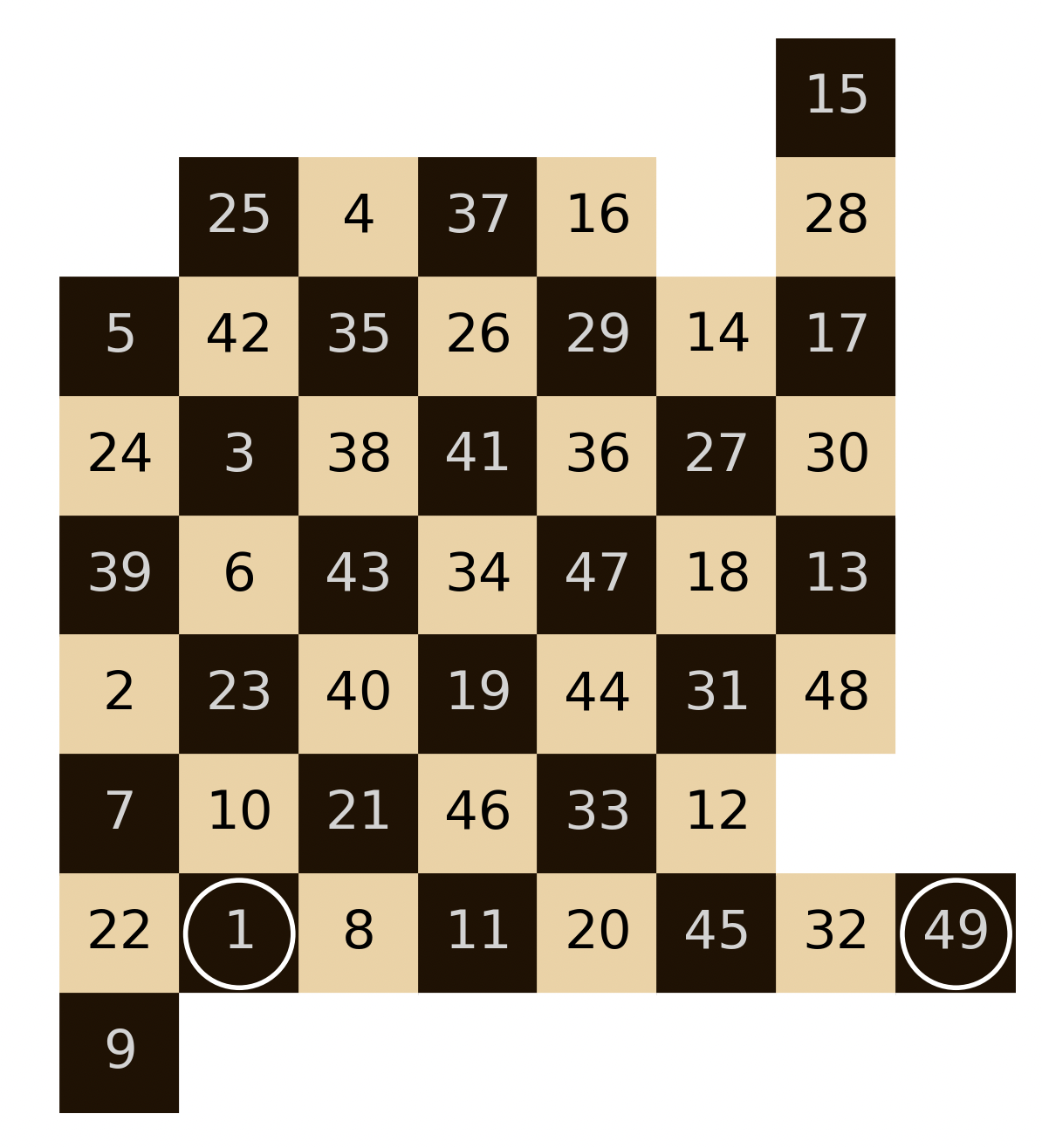} \caption{For every pair of transition squares on two 3-connected board up to symmetry, there is a knight's tour.} \label{3_connected_cases} \includegraphics[scale=0.23]{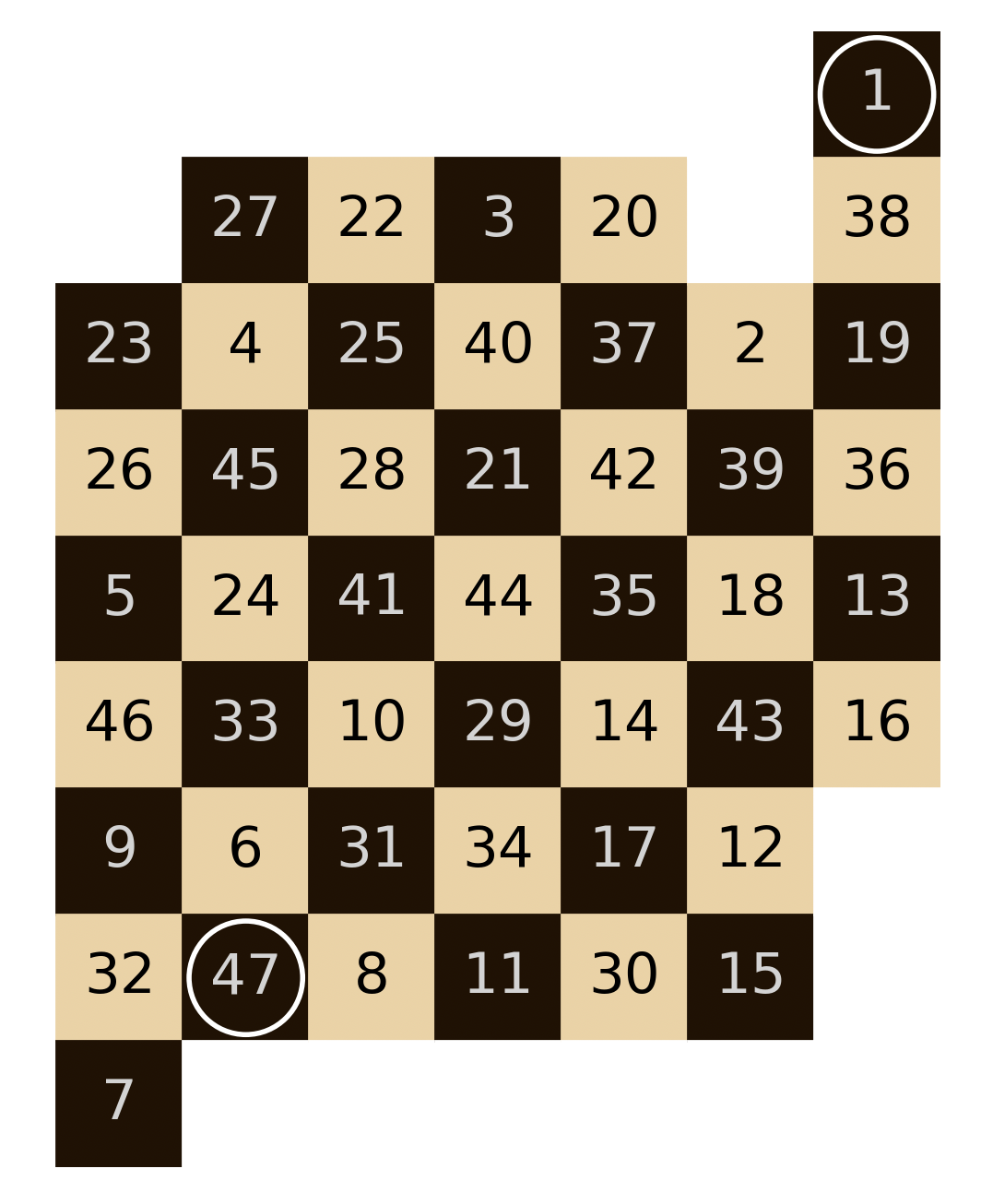} \includegraphics[scale=0.23]{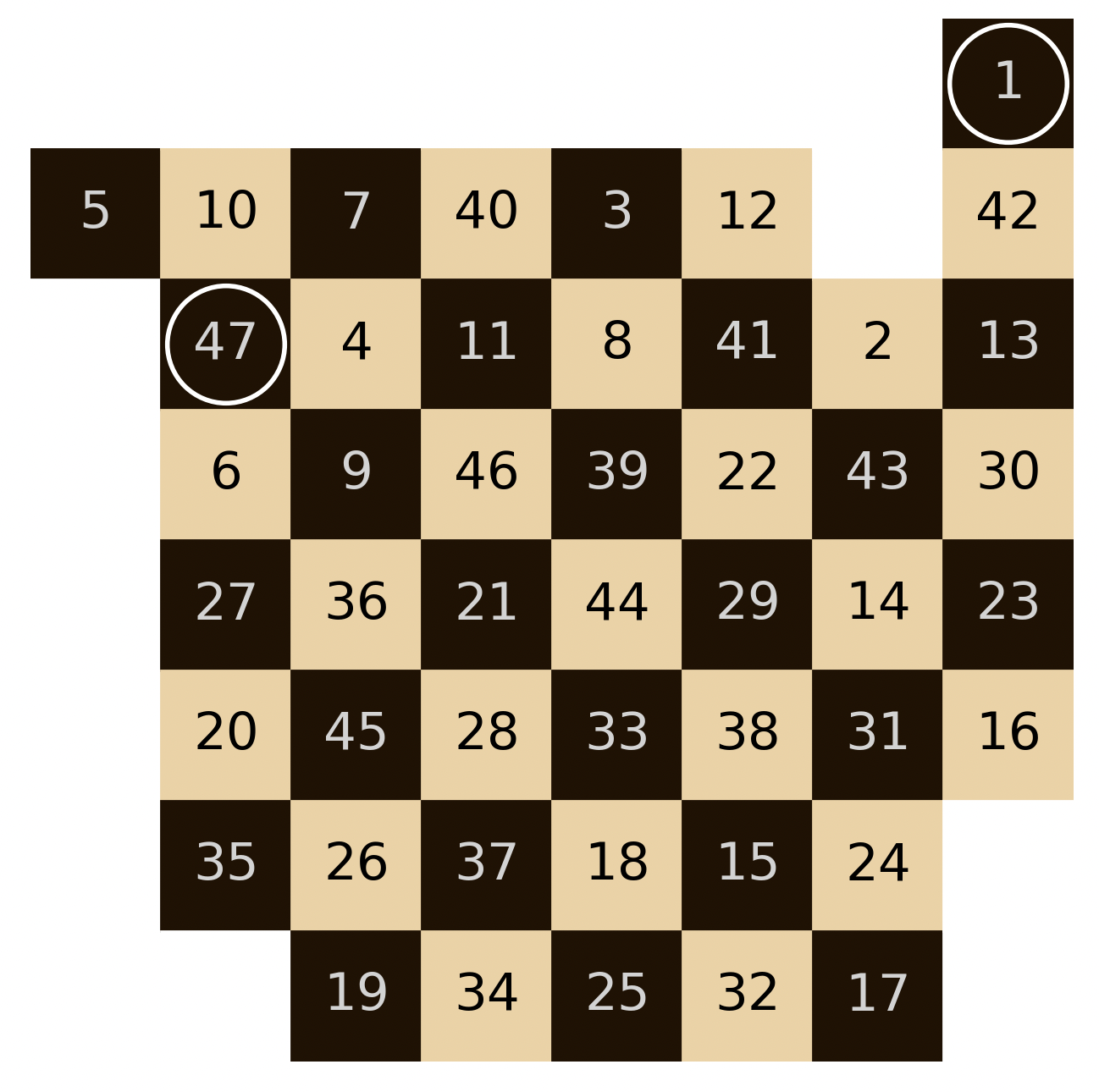} \includegraphics[scale=0.23]{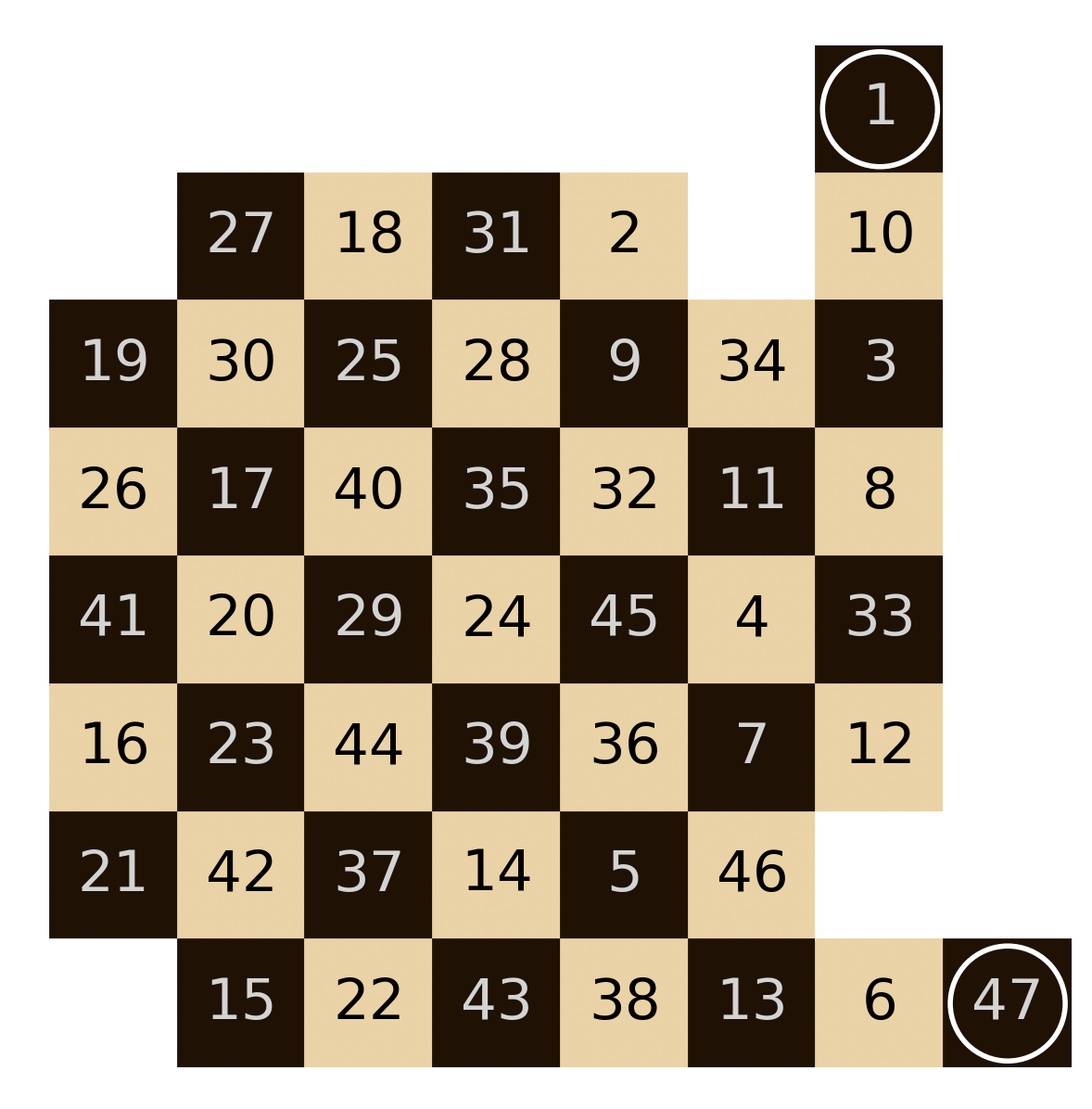} \caption{For every pair of transition squares on three 2-connected board up to symmetry, there is a knight's tour.} \label{2_connected_cases} \end{figure} 
\begin{figure} \centering \includegraphics[scale=0.23]{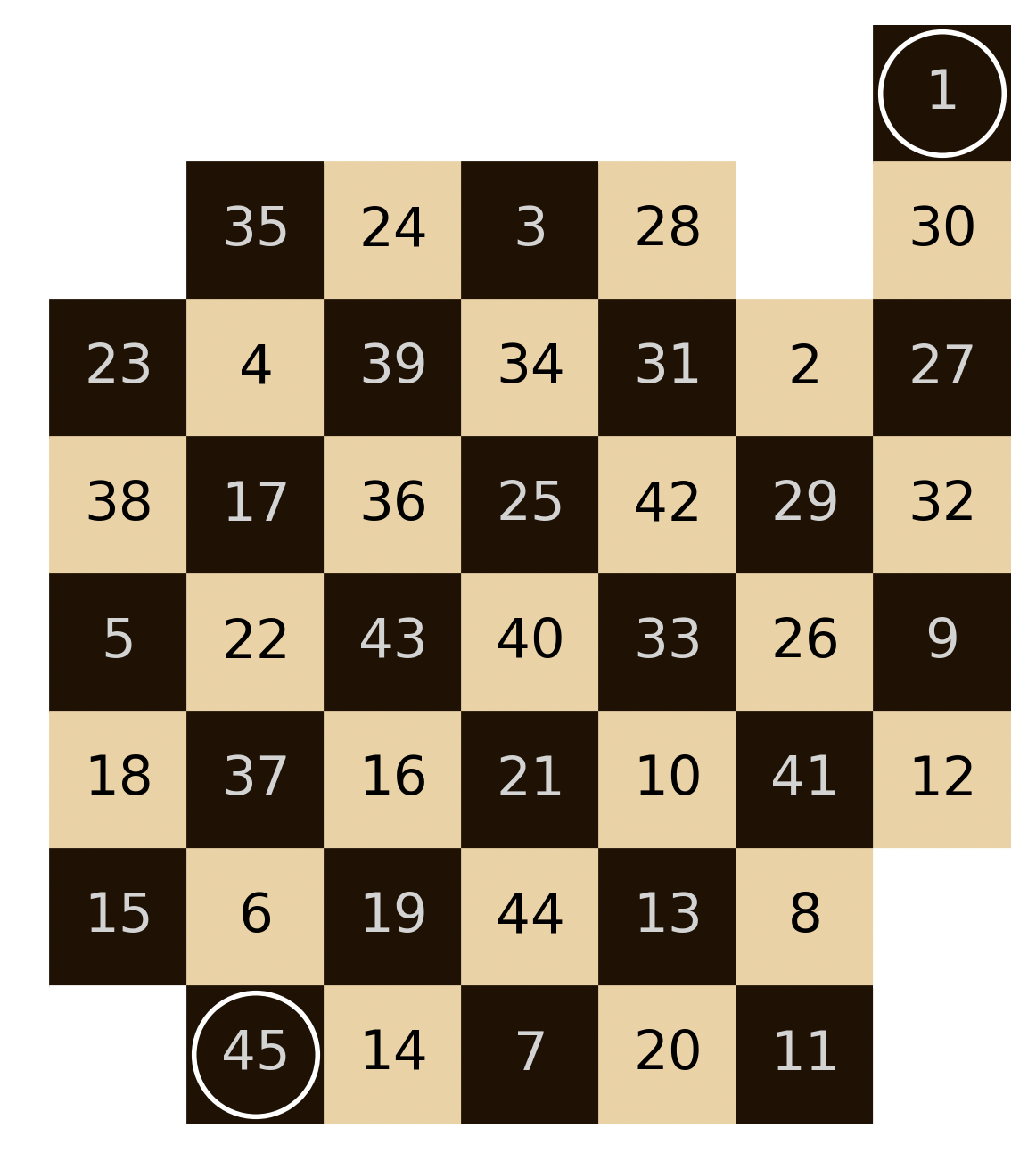} \includegraphics[scale=0.23]{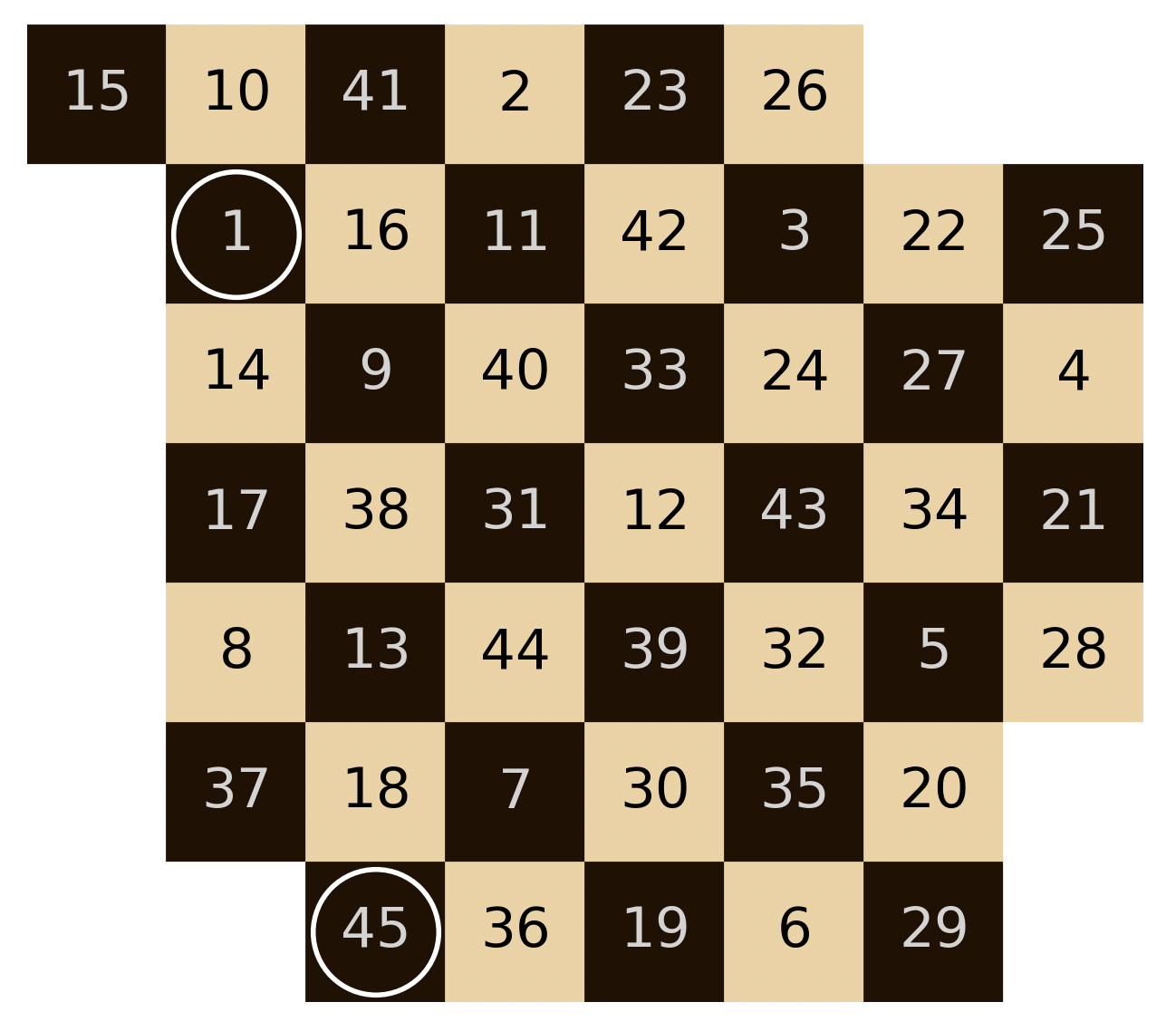} \caption{For both types of 1-connected boards, up to symmetry, there is a knight's tour starting on the transition square. For these boards, we only care about the starting square being a transition square because any knight tour that enters has to stay here.} 
\label{1-connected_cases} 

\end{figure}
    \begin{figure}[h]
    \centering
     \includegraphics[scale=0.3]{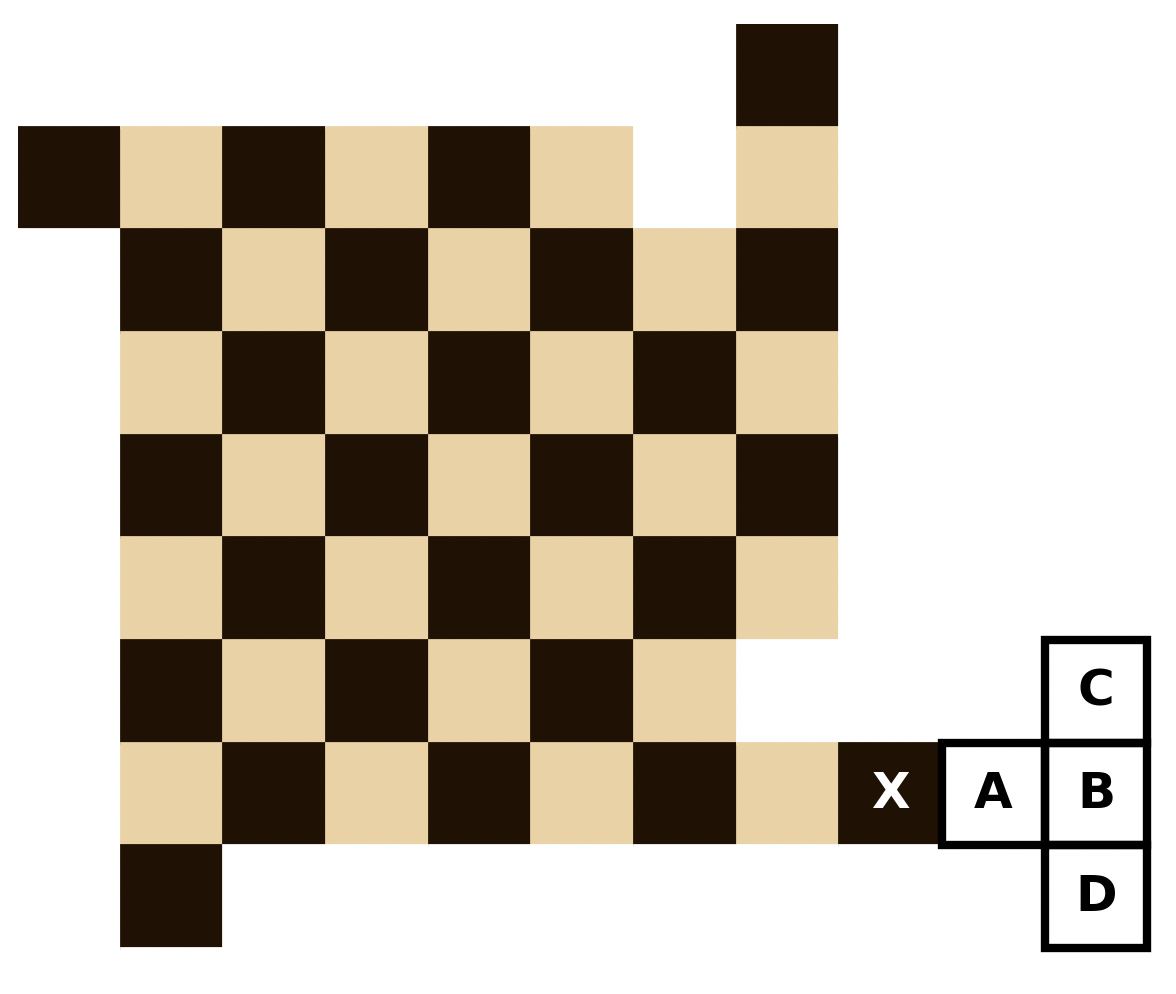}
    \caption{In our reduction to closed and open knight's tour, the chessboards coming from vertices on the right edge contain arm X, or a flipped version of it. We select a chessboard on the right edge and add squares A, B, C, and D. In this way, X is the only square that is adjacent to C and D in the knight's graph, while A and B are isolated vertices, so we do not place pebbles on A and B. The new chessboard is still connected. While $v_0$ was an arbitrary vertex in the pebble swap reduction, X is not, but identifying it can be done in poly-time.}
    \label{fig:reduction_is_extendable}
    \includegraphics[scale=0.25]{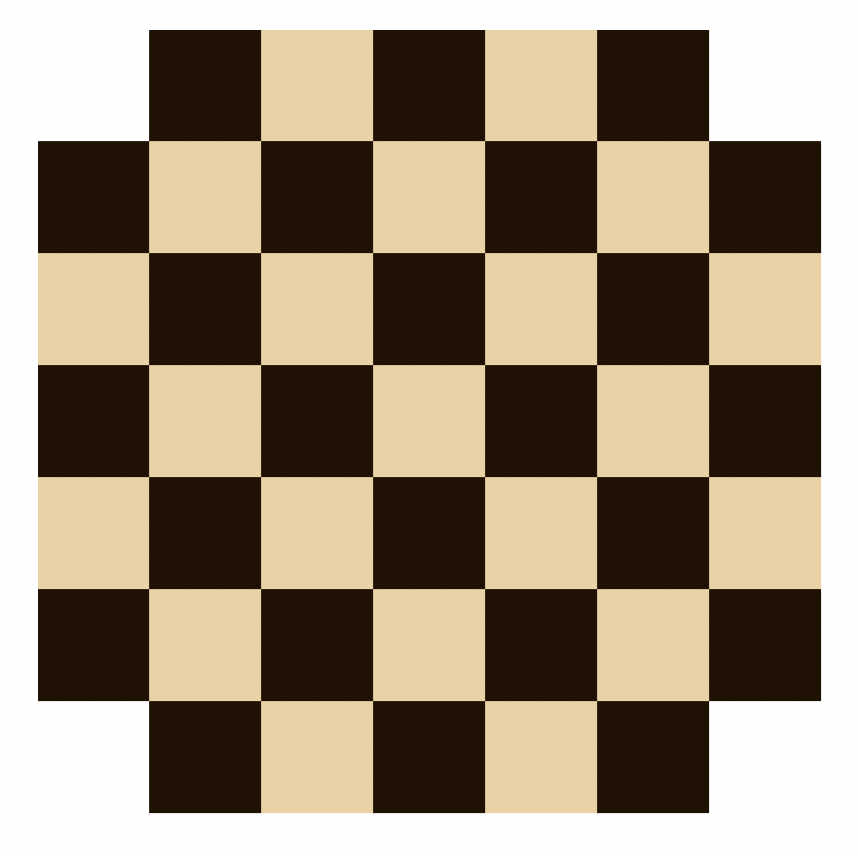}
    \caption{Shown above is an example of a chessboard that cannot be nicely extended. No matter how two additional squares are added, there are at least two squares in the original chessboard that are a knight's distance away from one of them. We don't build chessboards like this in our reduction, so this is not a problem. }
    \label{fig:not_extendable_example}
\end{figure}

\clearpage

\end{document}